\pdfoutput=1
\newcommand{\printProof}{}
\documentclass[sigconf,screen,nonacm]{acmart}
\newcommand{\tikzset}[1]{}
\newcounter{externalfigure}
\NewDocumentEnvironment{tikzpicture}{O{} +b}
  {\ifcase\value{externalfigure}%
     \includegraphics{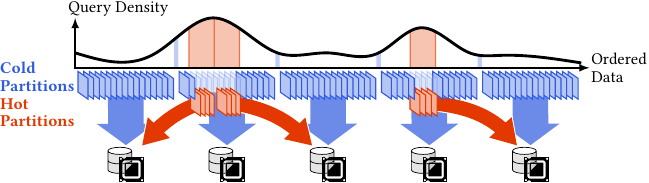}\or%
     \includegraphics{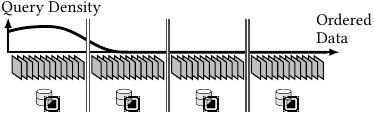}\or%
     \includegraphics{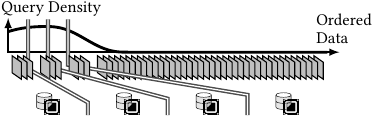}\or%
     \includegraphics{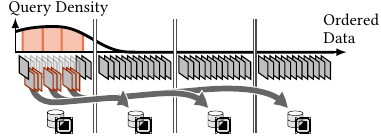}\or%
     \includegraphics{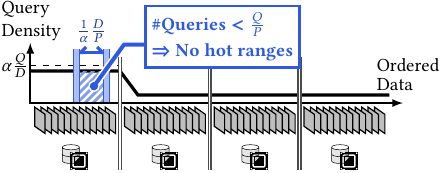}\or%
     \includegraphics{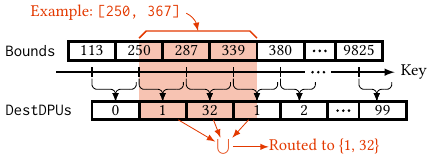}\or%
     \includegraphics{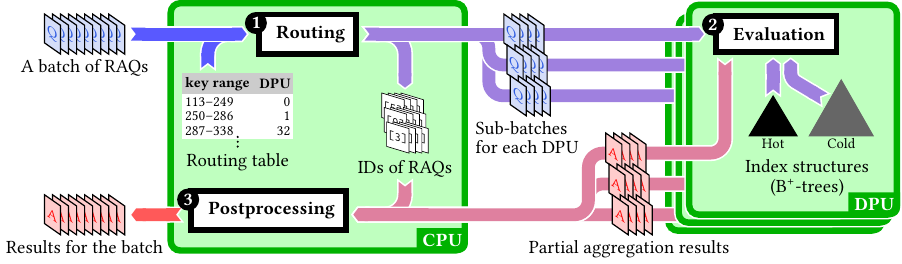}\else%
     \errmessage{tikz-external: more tikzpictures than external figures}\fi%
   \stepcounter{externalfigure}}
  {}

\AtBeginDocument{%
  }

\usepackage{amsmath}
\usepackage{amsfonts}
\usepackage{enumitem}
\usepackage{subcaption}
\usepackage{cleveref}
\usepackage{ifthen}
\usepackage{bm}
\usepackage{siunitx}
\usepackage{mathtools}
\usepackage{url}

\usepackage[ruled,lined,linesnumbered,nofillcomment,noend]{algorithm2e}
\DontPrintSemicolon
\SetKwInOut{Input}{Input}
\SetKwInOut{Output}{Output}
\SetKwRepeat{DoWhile}{do}{while}
\SetKwFunction{HotColdPartition}{HotColdPartition}
\SetKwFunction{FindHot}{FindHot}
\SetKwFunction{FindHotGreedy}{FindHotGreedy}
\SetKwFunction{RouteBatchedRAQ}{RouteBatchedRAQ}
\SetKwFunction{EvalSubBatch}{EvalSubBatch}
\SetKwFunction{PostProcessBatchedRAQ}{PostProcessBatchedRAQ}
\SetKwFunction{NQrys}{NQrys}
\SetKwFunction{SizeOf}{Size}
\SetKwFunction{KRng}{KRng}
\SetKwProg{Function}{Function}{}{end}
\renewcommand{\SetProgSty}[1]{\renewcommand{\ProgSty}[1]{\textnormal{\csname#1\endcsname{##1}}\unskip}}
\SetArgSty{textnormal}
\SetProgSty{textnormal}
\SetFuncArgSty{textnormal}
\SetCommentSty{textnormal}
\SetAlgoSkip{medskip}
\newcommand{\relmiddle}[1]{\mathrel{}\middle#1\mathrel{}}

\crefname{figure}{Fig.}{Figs.}
\Crefname{figure}{Figure}{Figures}
\crefname{table}{Table}{Tables}
\crefname{section}{Section}{Sections}
\crefname{appendix}{Appendix}{Appendices}
\crefname{algocf}{Algorithm}{Algorithms}
\crefname{theorem}{Theorem}{Theorems}
\crefname{lemma}{Lemma}{Lemmas}

\ifdefined\printProof%
\includecomment{Proof-in-Appendix}
\else%
\excludecomment{Proof-in-Appendix}
\fi%

\newcommand{\lineref}[1]{line~\ref{#1}}

\newcommand{\argmax}{\mathop{\arg\!\max}\limits}

\newlength{\notationsymwidth}

\newcommand{\BpForest}{B${}^\text{+}$-Forest\xspace}
\newcommand{\BpTree}{B${}^\text{+}$-tree\xspace}
\newcommand{\BpTrees}{B${}^\text{+}$-trees\xspace}
\newcommand{\PimTree}{PIM-tree\xspace}

\newcounter{extfootnote}

\definecolor{WhiteBlue}{HTML}{f9fdff}
\definecolor{LightBlue}{HTML}{dbf2ff}
\definecolor{Blue}{HTML}{2d5ae0}
\definecolor{DarkBlue}{HTML}{143447}

\definecolor{WhiteOrange}{HTML}{fffcf9}
\definecolor{LightOrange}{HTML}{ffeee0}
\definecolor{Orange}{HTML}{e33900}
\definecolor{DarkOrange}{HTML}{4f2604}

\definecolor{Purple}{HTML}{9400d3}
\definecolor{WhitePurple}{HTML}{faf0ff}
\definecolor{DarkPurple}{HTML}{441c58}

\definecolor{Green}{HTML}{009e73}
\definecolor{WhiteGreen}{HTML}{eefffa}
\definecolor{DarkGreen}{HTML}{144437}

\ifdefined\printProof
\else
\copyrightyear{2026}
\acmYear{2026}
\setcopyright{cc}
\setcctype{by}
\acmConference[ICPP '26]{Proceedings of the 55th International Conference on Parallel Processing}{September 28-October 01, 2026}{Singapore, Singapore}
\acmBooktitle{Proceedings of the 55th International Conference on Parallel Processing (ICPP '26), September 28-October 01, 2026, Singapore, Singapore}
\acmDOI{10.1145/3832810.3832916}
\acmISBN{979-8-4007-2657-6/2026/09}
\startPage{1}
\fi

\begin{document}


\author{Takato Hideshima}
\email{hideshima-takato182@g.ecc.u-tokyo.ac.jp}
\orcid{0009-0001-8078-3898}
\affiliation{%
  \institution{The University of Tokyo}
  \city{Tokyo}
  \country{Japan}
}

\author{Shigeyuki Sato}
\orcid{0000-0002-1496-1422}
\email{sato.shigeyuki@acm.org}
\affiliation{%
  \institution{The University of Electro-Communications}
  \city{Tokyo}
  \country{Japan}
}

\author{Tomoharu Ugawa}
\orcid{0000-0002-3849-8639}
\email{tugawa@acm.org}
\affiliation{%
  \institution{The University of Tokyo}
  \city{Tokyo}
  \country{Japan}
}


\tikzset{
 bbox/.style={inner sep=0pt, line width=0pt, fit={#1}},
 ellipsis/.style={dash pattern=on 0pt off 2*#1, dash phase=#1, line width=#1, line cap=round},
}

\title{Query Density-Driven Partitioning for Spatiotemporal Load Balancing on Processing-in-Memory Systems}

\begin{abstract}
Processing-in-Memory (PIM) systems, which consist of many processors
with small local memory, have recently emerged as commercial products
and attracted much attention as a means of overcoming the memory wall,
particularly in the context of in-memory database technology.
The state-of-the-art PIM-oriented index \PimTree has been demonstrated
to achieve asymptotically good spatiotemporal load balancing---query
loads and data sizes are balanced among processors---for skewed queries,
by trading spatial locality.  Unfortunately, such a sacrifice of spatial
locality hinders the PIM-oriented processing of range-aggregate queries.
To achieve both spatiotemporal load balancing and efficiently executing
range-aggregate queries on PIM systems, we develop a query
density-driven key-range partitioning scheme.  It balances query density
among PIM processors, allowing us to strike a balance between query load
and data size via a parameter.  We then develop \BpForest, a
PIM-oriented \BpTree variant based on our partitioning scheme.
 Experimental results demonstrated that it exhibits
higher skew resistance than a \BpTree based on
space-constrained, query-load-balanced, density-unaware partitioning,
and performance comparable to \PimTree in point-get queries, as well as
efficient support for range-aggregate queries.
\end{abstract}


\maketitle

\section{Introduction}

\newlength{\DPUSize}
\newcommand{\drawDPU}[3][\DPUSize]{
 \setlength{\DPUSize}{#1}
 \coordinate (tmp) at #2;
 \coordinate (memory bottom left) at ($(tmp) + (\DPUSize * -15/30, \DPUSize * -3/30)$);
 \coordinate (memory bottom right) at ($(memory bottom left) + (\DPUSize * 20/30, 0)$);
 \coordinate (memory top left) at ($(memory bottom left) + (0, \DPUSize * 16/30)$);

 \fill[white] let \p1 = ($(memory bottom right) - (memory bottom left)$),
                  \p2 = ($(memory top left) - (memory bottom left)$) in
  [x radius={\x1 / 2}, y radius = {\DPUSize * 3/30}]
   (memory top left) arc[start angle=180, end angle=-180];
 \fill[black!20] let \p1 = ($(memory bottom right) - (memory bottom left)$),
                     \p2 = ($(memory top left) - (memory bottom left)$) in
  [x radius={\x1 / 2}, y radius = {\DPUSize * 3/30}]
   (memory bottom right) arc[start angle=0, end angle=-60]
   -- ++(0, \y2)
   arc[start angle=-60, end angle=0]
   -- cycle;
 \fill[black!10] let \p1 = ($(memory bottom right) - (memory bottom left)$),
                     \p2 = ($(memory top left) - (memory bottom left)$) in
  [x radius={\x1 / 2}, y radius = {\DPUSize * 3/30}]
   (memory bottom left) arc[start angle=-180, end angle=-60]
   -- ++(0, \y2)
   arc[start angle=-60, end angle=-180]
   -- cycle;
 \draw[line width={\DPUSize * 0.02}] let \p1 = ($(memory bottom right) - (memory bottom left)$),
                                   \p2 = ($(memory top left) - (memory bottom left)$) in
  [x radius={\x1 / 2}, y radius = {\DPUSize * 3/30}]
   (memory top left) -- (memory bottom left)
   arc[start angle=-180, end angle=0]
   -- (memory top left -| memory bottom right)
   ($(memory bottom left)!.5!(memory top left)$) arc[start angle=-180, end angle=0]
   (memory top left) arc[start angle=-180, end angle=90] coordinate (memory north)
   arc[start angle=90, end angle=180];
 \node[inner sep=0pt] (core) at (memory bottom right) {\includegraphics[width={0.667\DPUSize}]{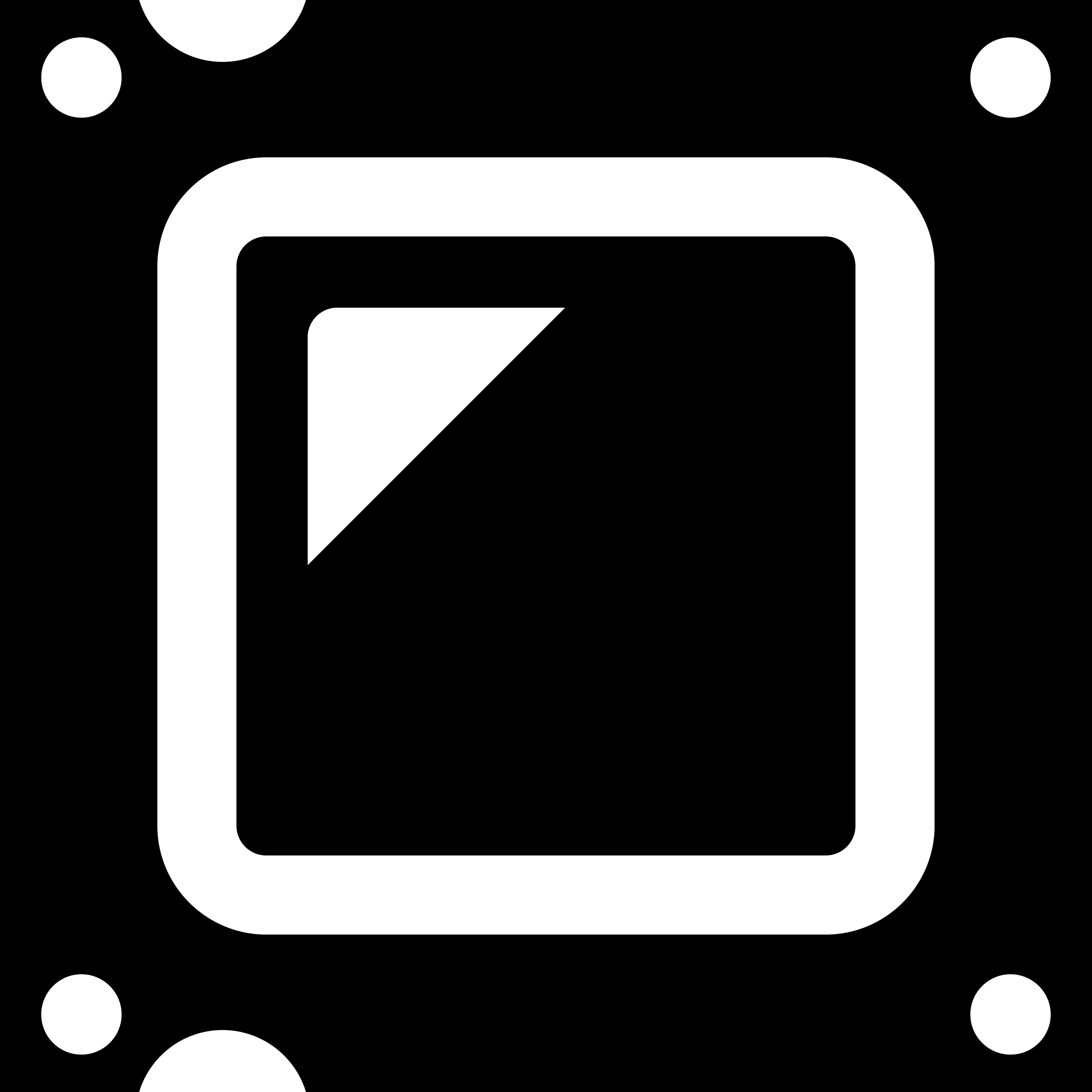}};

 \node[bbox={(memory north) (memory top left) (core.south east)}] (#3) {};
}
\newlength{\ItemWidth}
\newlength{\ItemHeight}
\newlength{\ItemSpacing}
\newlength{\BasePartSpacing}
\newlength{\HotPartSpacing}
\begin{figure*}[tb]
 \begin{minipage}[b]{0.7\linewidth}
  \centering
  \small
  \begin{tikzpicture}
   \setlength{\ItemWidth}{1.4ex}
   \setlength{\ItemHeight}{3.9ex}
   \setlength{\ItemSpacing}{0.8ex}
   \setlength{\BasePartSpacing}{2.2ex}
   \setlength{\HotPartSpacing}{1.8ex}

   \coordinate (key min);
   \coordinate (key max) at ($(key min) + (\ItemSpacing * 75 + \BasePartSpacing * 5, 0)$);
   \coordinate (load zero) at (key min);
   \coordinate (load max) at ($(load zero) + (0, 7ex)$);

   \coordinate (cold range0 key min) at (key min);
    \coordinate (cold range0 key max) at ($(key min)!.2!(key max)$);
   \coordinate (cold range1 key min) at (cold range0 key max);
    \coordinate (cold range1 key max) at ($(key min)!.4!(key max)$);
   \coordinate (cold range2 key min) at (cold range1 key max);
    \coordinate (cold range2 key max) at ($(key min)!.6!(key max)$);
   \coordinate (cold range3 key min) at (cold range2 key max);
    \coordinate (cold range3 key max) at ($(key min)!.8!(key max)$);
   \coordinate (cold range4 key min) at (cold range3 key max);
    \coordinate (cold range4 key max) at (key max);

   \coordinate (hot0 begin) at ($(cold range1 key min)!.125!(cold range1 key max)$);
    \coordinate (hot0 end) at ($(cold range1 key min)!.375!(cold range1 key max)$);
   \coordinate (hot1 begin) at (hot0 end);
    \coordinate (hot1 end) at ($(cold range1 key min)!.625!(cold range1 key max)$);
   \coordinate (hot2 begin) at ($(cold range3 key min)!.3125!(cold range3 key max)$);
    \coordinate (hot2 end) at ($(cold range3 key min)!.5625!(cold range3 key max)$);

   \newcommand{\plotQueryDist}[2][]{
    #2
     let \p1 = (key min), \p2 = ($(key max |- load max) - (key min)$) in
     plot[samples=100, variable=\t, domain={-0.001:0.161}] (\t * \x2 + \x1, {(+118.922 * \t*\t*\t +0.000 * \t*\t -3.044 * \t +0.300) * \y2 + \y1}) #1
     plot[samples=100, variable=\t, domain={0.159:0.271}] (\t * \x2 + \x1, {(-496.223 * \t*\t*\t +295.270 * \t*\t -50.288 * \t +2.820) * \y2 + \y1}) #1
     plot[samples=100, variable=\t, domain={0.269:0.401}] (\t * \x2 + \x1, {(+464.405 * \t*\t*\t -482.839 * \t*\t +159.802 * \t -16.088) * \y2 + \y1}) #1
     plot[samples=100, variable=\t, domain={0.399:0.501}] (\t * \x2 + \x1, {(-388.981 * \t*\t*\t +541.224 * \t*\t -249.823 * \t +38.528) * \y2 + \y1}) #1
     plot[samples=100, variable=\t, domain={0.499:0.601}] (\t * \x2 + \x1, {(+455.957 * \t*\t*\t -726.183 * \t*\t +383.880 * \t -67.089) * \y2 + \y1}) #1
     plot[samples=100, variable=\t, domain={0.599:0.681}] (\t * \x2 + \x1, {(-969.951 * \t*\t*\t +1840.452 * \t*\t -1156.101 * \t +240.907) * \y2 + \y1}) #1
     plot[samples=100, variable=\t, domain={0.679:0.781}] (\t * \x2 + \x1, {(+742.694 * \t*\t*\t -1653.343 * \t*\t +1219.680 * \t -297.603) * \y2 + \y1}) #1
     plot[samples=100, variable=\t, domain={0.779:0.841}] (\t * \x2 + \x1, {(-537.785 * \t*\t*\t +1342.978 * \t*\t -1117.451 * \t +310.051) * \y2 + \y1}) #1
     plot[samples=100, variable=\t, domain={0.839:1.001}] (\t * \x2 + \x1, {(+25.502 * \t*\t*\t -76.507 * \t*\t +74.917 * \t -23.812) * \y2 + \y1}) #1
    ;
   }
   \begin{scope}
    \clip (hot0 begin) rectangle (hot1 end |- load max)
          (hot2 begin) rectangle (hot2 end |- load max);
    \plotQueryDist[coordinate (tmp) -- (tmp |- key min) -| cycle]{\fill[Orange!30]}
   \end{scope}
   \begin{scope}
    \clip ($(hot0 begin) + (-0.1ex, 0)$) rectangle ($(hot0 begin |- load max) + (0.1ex, 0)$)
          ($(hot1 begin) + (-0.1ex, 0)$) rectangle ($(hot1 begin |- load max) + (0.1ex, 0)$)
          ($(hot1 end) + (-0.1ex, 0)$) rectangle ($(hot1 end |- load max) + (0.1ex, 0)$)
          ($(hot2 begin) + (-0.1ex, 0)$) rectangle ($(hot2 begin |- load max) + (0.1ex, 0)$)
          ($(hot2 end) + (-0.1ex, 0)$) rectangle ($(hot2 end |- load max) + (0.1ex, 0)$);
    \plotQueryDist[coordinate (tmp) -- (tmp |- key min) -| cycle]{\fill[Orange!70]}
   \end{scope}
   \begin{scope}
    \clip ($(cold range0 key max) + (-0.3ex, 0)$) rectangle ($(cold range0 key max |- load max) + (0.3ex, 0)$)
          ($(cold range1 key max) + (-0.3ex, 0)$) rectangle ($(cold range1 key max |- load max) + (0.3ex, 0)$)
          ($(cold range2 key max) + (-0.3ex, 0)$) rectangle ($(cold range2 key max |- load max) + (0.3ex, 0)$)
          ($(cold range3 key max) + (-0.3ex, 0)$) rectangle ($(cold range3 key max |- load max) + (0.3ex, 0)$)
          ($(cold range4 key max) + (-0.3ex, 0)$) rectangle ($(cold range4 key max |- load max) + (0.3ex, 0)$);
    \plotQueryDist[coordinate (tmp) -- (tmp |- key min) -| cycle]{\fill[Blue!30]}
   \end{scope}
   \plotQueryDist{\draw[line width=0.4ex]}

   \draw[line width=0.2ex, -latex, line cap=rect]
    (key min) -- ($(key max) + (1.2ex, 0)$) coordinate (tmp);
   \node[inner sep=0.2ex, anchor=west] (item axis label) at (tmp) {\shortstack[l]{Ordered\\Data}};

   \draw[line width=0.2ex, -latex, line cap=rect]
    (load zero) -- (load max);
   \node[inner sep=0.2ex, anchor=south west] (query axis label) at ($(load max) + (-1.2ex, 0)$) {Query Density};

   \coordinate (cold range height) at ($(key min) + (0, -0.5ex)$);
   \coordinate (hot range height) at ($(cold range height) + (0, -\ItemHeight + \ItemWidth * 0.7 + \ItemSpacing * 0.7)$);

   \foreach \i in {0, ..., 4} {
    \coordinate (cold range\i\space begin) at ($(cold range height -| cold range\i\space key min) + (\ItemSpacing * 0.5, 0)$);
     \coordinate (cold range\i\space end) at ($(cold range height -| cold range\i\space key max) + (\ItemSpacing * 1.5 - \BasePartSpacing, 0)$);
   }
   \coordinate (hot range0 begin) at ($(hot range height -| cold range1 begin)!.125!(hot range height -| cold range1 end)$);
    \coordinate (hot range0 end) at ($(hot range height -| cold range1 begin)!.375!(hot range height -| cold range1 end) + (\ItemSpacing * 2 / 3 - \HotPartSpacing * 2 / 3, 0)$);
   \coordinate (hot range1 begin) at ($(hot range height -| cold range1 begin)!.375!(hot range height -| cold range1 end) + (\HotPartSpacing / 2 - \ItemSpacing / 2, 0)$);
    \coordinate (hot range1 end) at ($(hot range height -| cold range1 begin)!.625!(hot range height -| cold range1 end) + (\ItemSpacing / 6 - \HotPartSpacing / 6, 0)$);
   \coordinate (hot range2 begin) at ($(hot range height -| cold range3 begin)!.3125!(hot range height -| cold range3 end)$);
    \coordinate (hot range2 end) at ($(hot range height -| cold range3 begin)!.5625!(hot range height -| cold range3 end)$);

   \foreach \i in {0, ..., 4} {
    \coordinate (cold\i\space assign base) at ($(cold range\i\space key min |- hot range height)!.5!(cold range\i\space key max |- hot range height) + (0, \ItemSpacing * -0.7)$);

    \drawDPU[5ex]{($(cold\i\space assign base) + (0, -\ItemHeight - 2.5ex - 4ex)$)}{proc\i}

    \begin{scope}[on behind layer]
     \draw[Blue!70, line width=5ex, -{Triangle[width=8ex, length=3ex]}] (cold\i\space assign base) -- ($(proc\i.north) + (0, 0.3ex)$);
    \end{scope}
   }

   \begin{scope}[
    cold item/.style={shift={#1}, yslant=-0.7, draw=Blue, fill=Blue!35, line width=0.15ex},
    hot item/.style={shift={#1}, yslant=-0.7, draw=Orange, fill=Orange!40, line width=0.15ex},
   ]
    \foreach \i in {15, ..., 0} {
     \path[cold item={($(cold range4 begin)!{1.0 / 16 * \i}!(cold range4 end)$)}] (0, 0) rectangle (\ItemWidth, -\ItemHeight + \ItemWidth * 0.7);
    }

    \foreach \i in {15, ..., 9} {
     \path[cold item={($(cold range3 begin)!{1.0 / 16 * \i}!(cold range3 end)$)}] (0, 0) rectangle (\ItemWidth, -\ItemHeight + \ItemWidth * 0.7);
    }
    \foreach \i in {8, ..., 5} {
     \path[cold item={($(cold range3 begin)!{1.0 / 16 * \i}!(cold range3 end)$)}, opacity=0.25] (0, 0) rectangle (\ItemWidth, -\ItemHeight + \ItemWidth * 0.7);
    }
    \foreach \i in {3, ..., 0} {
     \path[hot item={($(hot range2 begin)!{1.0 / 4 * \i}!(hot range2 end)$)}] (0, 0) rectangle (\ItemWidth, -\ItemHeight + \ItemWidth * 0.7)
      coordinate (hot2 item\i);
    }
    \foreach \i in {4, ..., 0} {
     \path[cold item={($(cold range3 begin)!{1.0 / 16 * \i}!(cold range3 end)$)}] (0, 0) rectangle (\ItemWidth, -\ItemHeight + \ItemWidth * 0.7);
    }

    \foreach \i in {15, ..., 0} {
     \path[cold item={($(cold range2 begin)!{1.0 / 16 * \i}!(cold range2 end)$)}] (0, 0) rectangle (\ItemWidth, -\ItemHeight + \ItemWidth * 0.7);
    }

    \foreach \i in {15, ..., 10} {
     \path[cold item={($(cold range1 begin)!{1.0 / 16 * \i}!(cold range1 end)$)}] (0, 0) rectangle (\ItemWidth, -\ItemHeight + \ItemWidth * 0.7);
    }
    \foreach \i in {9, ..., 2} {
     \path[cold item={($(cold range1 begin)!{1.0 / 16 * \i}!(cold range1 end)$)}, opacity=0.25] (0, 0) rectangle (\ItemWidth, -\ItemHeight + \ItemWidth * 0.7);
    }
    \foreach \i in {3, ..., 0} {
     \path[hot item={($(hot range1 begin)!{1.0 / 4 * \i}!(hot range1 end)$)}] (0, 0) rectangle (\ItemWidth, -\ItemHeight + \ItemWidth * 0.7)
      coordinate (hot1 item\i);
    }
    \foreach \i in {3, ..., 0} {
     \path[hot item={($(hot range0 begin)!{1.0 / 4 * \i}!(hot range0 end)$)}] (0, 0) rectangle (\ItemWidth, -\ItemHeight + \ItemWidth * 0.7)
      coordinate (hot0 item\i);
    }
    \foreach \i in {1, ..., 0} {
     \path[cold item={($(cold range1 begin)!{1.0 / 16 * \i}!(cold range1 end)$)}] (0, 0) rectangle (\ItemWidth, -\ItemHeight + \ItemWidth * 0.7);
    }

    \foreach \i in {15, ..., 0} {
     \path[cold item={($(cold range0 begin)!{1.0 / 16 * \i}!(cold range0 end)$)}] (0, 0) rectangle (\ItemWidth, -\ItemHeight + \ItemWidth * 0.7);
    }
   \end{scope}

   \foreach \i in {0, ..., 2} {
    \node[bbox={(hot range\i\space begin) (hot\i\space item3)}] (hot range\i) {};
   }
   \begin{scope}[
    on behind layer,
    hot assignment/.style={Orange, line width=2ex, -{Triangle[width=4ex, length=3.464ex, flex=1]}, shorten >= 0.3ex,
                           preaction={draw=white, line width=2.3ex, -{Triangle[width=4.520ex, length=3.914ex, flex=1]}, shorten <= -0.15ex, shorten >= 0ex}}
   ]
    \draw[hot assignment, out=200, in=50] (hot range0.center) to (proc0);
    \draw[hot assignment, out=-10, in=130] (hot range1.center) to (proc2);
    \draw[hot assignment, out=-5, in=125] (hot range2.center) to (proc4);
   \end{scope}

   \begin{scope}[inner sep=0pt, font=\bfseries]
    \node[text=Orange, anchor=north east] (hot range label) at ($(hot range height -| key min) + (0, -1.4ex)$) {\shortstack[l]{
     Hot \\
     Partitions}};
    \node[text=Blue, anchor=south west] at ($(hot range label.north west) + (0, 0.9ex)$) {\shortstack[l]{
     Cold \\
     Partitions}};
   \end{scope}
  \end{tikzpicture}
  \caption{Overview of query density-driven partitioning.}
  \label{fig:hot/cold-part}
  \end{minipage}
 \begin{minipage}[b]{0.28\linewidth}
  \centering
  \includegraphics[width=0.7\linewidth]{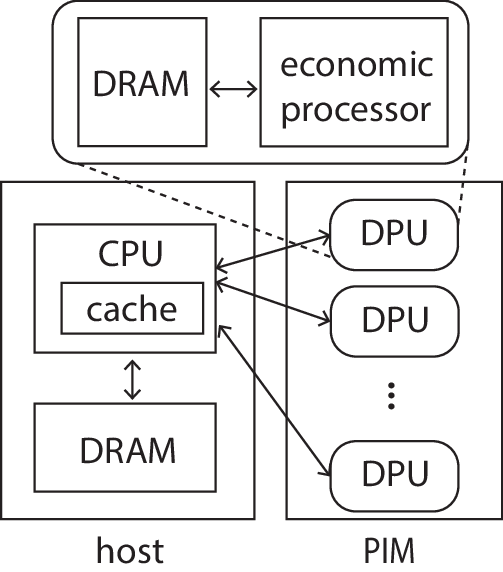}
  \caption{Structure of UPMEM PIM.}
  \label{fig:UPMEM}
 \end{minipage}
\end{figure*}

The ``memory wall'' problem has become increasingly challenging for
data-intensive computing because physical data movement incurs
inevitable costs.  One natural, promising approach to
breaking the memory wall is to distribute data across the small local memory
of many processors and process it in parallel with reduced data
movement.  This idea is commonly shared by Massively Parallel Processor
Array (MPPA)~\cite{butts07:MPPA,dupont15:Kalray,hatta25:PEZY-SC4s} and
Processing-in-Memory (PIM)~\cite{asifuzzaman23:survey_pim}
architectures, which tackle the memory wall.

Since the first real PIM system UPMEM PIM~\cite{UPMEM} emerged, a
variety of PIM
applications~\cite{nider21:UPMEM,gomez-luna22:PrIM,falevoz24:energy_efficiecy}
have been studied.  One of the most promising applications is in-memory
databases~\cite{kepe19:db_pim,bernhardt23:pimDB}, particularly key-value
stores, for two major reasons.  They are easy to exploit parallelism by
batching queries as in recent batch-parallel data
structures~\cite{wang22:ParGeo,wheatman24:CPMA,kang21:PIM_model,kang22:PIM-tree}.
Computation with database indexes is data-driven; i.e., computations
arise only at relevant parts of the entire data.
For in-memory key-value stores on UPMEM PIM, various PIM-oriented index
data
structures~\cite{choe19:cds_ndp,choe22:HybriDS,kang21:PIM_model,kang22:PIM-tree,kang23:PIM-trie,hua24:RADAR,cui25:PIMLex}
were developed, supporting point and range queries.

The main challenge in in-memory key-value stores on UPMEM PIM, or many
processors with small local memory, has been resistance to skewed
queries, where data accesses of workloads are non-uniform, biased as in
real-world workloads~\cite{cao20:RocksDB_workloads}.  To fully leverage
systems like UPMEM PIM, not only query loads on processors but also the
data sizes assigned to them must be balanced even under skewed queries.
We call this requirement \emph{spatiotemporal} load balancing.

The state-of-the-art data structure
\PimTree~\cite{kang22:PIM-tree} has achieved asymptotically good
spatiotemporal load balancing.
A key trick of \PimTree for spatiotemporal load balancing is random
scattering of small (specifically, cache line-granular) chunks of the
key space over PIM processors.  It works asymptotically very well for
point queries and range-scan queries because queries to the same chunk
can be fused, and frequently accessed chunks are evenly distributed with
high probability.  The resultant spatiotemporal balance satisfies the
resource constraints of UPMEM PIM.

Unfortunately, random chunk scattering sacrifices spatial
locality.  Range queries with long key ranges result in lying over many
processors, thereby amplifying the communication between CPUs and PIM
processors.  Such a sacrifice of spatial locality becomes more severe for range-aggregate
queries because ranges that fit into a single processor can be
efficiently aggregated with minimal communication, which is the 
PIM-friendly situation.  To make matters
worse, query fusion, which is feasible for range-scan queries,
does not work well with user-defined aggregate
functions~\cite{cohen06:udaf,foufoulas23:udf}.
 It thus remains an open challenge to achieve spatiotemporal load balance
under query skew while preserving the spatial locality required for
efficient range-aggregate query processing.

To this end, we present a novel workload-driven key-range partitioning
scheme for ordered-key indexes.  It partitions the key range to balance
\emph{query density}, i.e., the number of hit queries per data amount,
across PIM processors.  High-density, i.e., \emph{hot}, ranges are
identified and distributed so that each PIM processor has at most one
hot range and the key-value pairs therein, as illustrated in
\cref{fig:hot/cold-part}.  The remaining ranges, i.e., cold ranges, and
the key-value pairs therein are assigned to PIM processors in a
block-partitioning manner.  Supposing that skewed queries make hot
ranges few and narrow, a dominant part of the data is distributed
contiguously and mostly evenly, preserving spatial locality, while query
loads are balanced.  Our partitioning algorithm also has a tuning
parameter that allows us to find a sweet spot of the trade-off between
query load balance and data size balance empirically.  Moreover, our query
density-driven approach is index-agnostic and applicable to any
ordered-key index.

Then, we develop \BpForest, an adaptation of our partitioning algorithm
to \BpTrees on top of UPMEM PIM by following the PIM-oriented hybrid
design~\cite{choe22:HybriDS}, where the upper part of the \BpTree
resides as a routing table on the CPU side, and the lower part, i.e., a
forest of subtrees, is distributed to the PIM side.  Our algorithm for
executing batched general range-aggregate queries on \BpTrees follows
this structure straightforwardly; it performs query routing on the CPU
side by following the routing table, query execution for subtrees on the
PIM side, and the final aggregation on the CPU side.  It is
simple yet efficiently meets the resource constraints of UPMEM PIM,
thanks to the high degree of spatial locality of our partitioning
scheme.

We experimentally evaluate \BpForest against
 two baselines: a distributed \BpTree based
on query-load-balanced yet density-unaware variable-sized block
partitioning (i.e., classic chains-on-chains
partitioning~\cite{pinar04:optimal_1d_part}) under space constraints; and
\PimTree, the state-of-the-art PIM-oriented index.
\BpForest has achieved more robust load balancing under space
constraints for skewed queries than the density-unaware partitioned
\BpTree, demonstrating the effectiveness of query density-driven
partitioning for skew resistance.  \BpForest has achieved performance
comparable to \PimTree in an apples-to-apples comparison on point-get
queries, yet has been able to handle range-aggregate queries
efficiently, unlike \PimTree, and benefit from larger batches of
queries.

Our main contributions are summarized as follows:
\begin{itemize}
 \item We have presented a query-driven key-range partitioning scheme
       for ordered-key indexes, for spatiotemporal load balancing
       (\cref{sec:load_balance}).  It balances query
       density among PIM processors by distributing high query-density
       ranges, offering a tuning parameter to control the balance of
       query load and data size.

 \item We have developed \BpForest by applying our partitioning scheme
       to \BpTrees, following the PIM-oriented hybrid
       design~\cite{choe22:HybriDS} on top of UPMEM PIM
       (\cref{sec:bp-forest}).  Our algorithm for batched
       range-aggregate queries to \BpForest is simple yet efficient,
       thanks to the high degree of spatial locality of our partitioning
       scheme.

 \item We have experimentally demonstrated that \BpForest exhibits higher
       skew resistance than a distributed \BpTree based on
       query-load-balanced yet density-unaware variable-sized block
       partitioning~\cite{pinar04:optimal_1d_part} (\cref{sec:eval}). It also performs
       comparably to the state-of-the-art PIM-oriented index
       \PimTree~\cite{kang22:PIM-tree} in point-get queries, despite its ability to handle
       range-aggregate queries efficiently.
\end{itemize}

\subsubsection*{Scope and Limitations}
This paper assumes offline full rebalancing based on the presented partitioning
scheme for brevity.  A sophisticated online partial rebalancing method is
addressed in our concurrent paper~\cite{hideshima26:partial_resharding}
for dynamic workload changes.
Our approach does not depend on UPMEM PIM and is naturally
amenable to a wide range of MPPA architectures, whereas \BpForest is
not feasible for lockstep PIM architectures, such as
HBM-PIM~\cite{lee21:HBM-PIM}.

\section{UPMEM PIM}
\label{sec:upmem}

The computational model of UPMEM PIM consists of a host and a large
number of PIM processors called DPUs operating in parallel, as illustrated
in \cref{fig:UPMEM}.  Each DPU consists of only 64~MiB DRAM and an
economic \SI{400}{MHz} processor.  UPMEM PIM has a conventional
heterogeneous configuration: the host utilizes the DPUs as accelerators.
The CPU and all DPUs do not share memory, and every
communication is explicitly initiated by the CPU and allowed only
between the CPU and the DPUs.  In typical use cases, DPUs perform
homogeneous computations, i.e., the same kernel code is simultaneously
invoked for all DPUs.

In the computational model of UPMEM PIM, the CPU and DPUs have different
strengths.  The CPU is advantageous to
high-locality computations owing to the processor cache; the DPUs are
advantageous for low-locality parallel computations because data
accesses dominate the total time, and their throughput benefits from
parallel accesses.  A practical choice of concurrency control is only
barrier synchronization such that the CPU waits for all DPUs to
terminate.  See \cite{gomez-luna22:PrIM} for the details.

Considering the above, spatiotemporal load balancing is necessary to
fully leverage UPMEM PIM.
Spatial load balance is required to handle
such large-scale data that is suitable for parallel processing on DPUs,
because of low-capacity DRAM per DPU.
Temporal load balance among DPUs is also required
 to achieve the scalability derived from the
parallelism of DPUs.

\section{Batched Range-Aggregate Queries}
\label{sec:range-aggregate}

This paper deals with user-defined range-aggregate queries (RAQs). An
RAQ consists of a key range $r_q$ and an aggregator $(f,\oplus)$ and
calculates $\bigoplus f(k,v)$ for all key-value pairs $(k,v)$ such that
$k \in r_q$. As an instance of RAQs, a range-count-if query counts items
in a specified key range that meet a specified condition: with an
aggregator $(f, +)$ such that $f(k, 2i) = 1$ and $f(k, 2i+1) = 0$ for $i
\in \mathbb{Z}$, it counts the number of items whose values are even.
For aggregator $(f,\oplus)$, we call $\oplus$ the reduction operator.
We assume that the reduction operator is associative and has an
identity.  RAQs are a generalization of range-scan queries, where a
query of range $r_q$ returns all key-value pairs $(k,v)$ such that $k
\in r_q$.

We address batch-parallel queries, i.e., parallel processing of a batch
of RAQs.  It is widely used as a representation of parallel tasks for
various data structures, such as k-d trees~\cite{wang22:ParGeo}, skip
lists~\cite{kang21:PIM_model,kang22:PIM-tree}, array-based
B-trees~\cite{wheatman24:CPMA}, and tries~\cite{kang23:PIM-trie}.
Following these prior studies, we assume a batch to consist of
RAQs of the same type.

Key-range partitioning is a common approach to load balancing of
distributed key-value stores.  It divides key-value pairs into
non-overlapping sets (i.e., partitions) over the key space, and
distributes them over workers, providing key-range assignments.  Each
partition typically stores key-value pairs in key order.  Queries are
dispatched to the workers according to the
key-range assignments.  For point queries, each worker can execute
dispatched queries as-is within its own partitions, resulting in
high-performance query processing.  For RAQs, each worker can similarly
execute given queries as-is, except that the postprocessing of partial
aggregation results of RAQs spanning multiple workers is required.

We assume that the number of items involved in the key range of an RAQ
is about 100, according to the report~\cite{cao20:RocksDB_workloads},
where most range accesses in the real-world workloads at Facebook
involved fewer than 100 items.  Long-range queries have distinct
effects on load balancing because they can span many workers, amplifying
the queries.  Handling range queries uniformly regardless of range
length is impractical.  We focus on skewed range queries with a typical
range length.

Although coarser-grained key-range partitioning suppresses range query
amplification more, UPMEM PIM requires spatiotemporal load balancing.
Meeting both at once is still an open problem.

\section{Query Density-Driven Partitioning}
\label{sec:load_balance}

\subsection{Key Ideas}
\label{sec:load_balance:idea}

Our primary observation is that  \emph{query density}, i.e., the number of hit
queries per data size for given queries and data items, is very
important for spatiotemporal load balancing for skewed queries.  A significant
gap in query density among DPUs translates to an imbalance in query load
or data size.

A gap in query density among DPUs can be algorithmically suppressed by
partitioning a given dataset into fine-grained pieces and scattering
them randomly over DPUs, as in existing PIM-oriented data
structures~\cite{kang21:PIM_model,kang22:PIM-tree}, because \emph{hot}
ranges of high density and \emph{cold} ranges of low density will be
mixed within the data assigned to each DPU.  Unfortunately, as mentioned
in \cref{sec:range-aggregate}, fine-grained partitioning is not suitable
for batched RAQs.

Our key idea to achieve both coarse-grained partitioning and query
density balance is to pinpoint, divide, and distribute hot ranges and
key-value pairs therein over DPUs, as illustrated in
\cref{fig:example-partitions:hot_cold}.  Now, we consider assigning
 a small number of partitions to each DPU because single-partition
schemes, as in classic one-dimensional
partitioning~\cite{pinar04:optimal_1d_part,lieber14:scalable_1d_part},
are difficult to achieve even density in the presence of high skewness,
as illustrated in \cref{fig:example-partitions}.

\begin{figure}[tb]
 \small
 \setlength{\ItemWidth}{1.3ex}
 \setlength{\ItemHeight}{3.3ex}
 \setlength{\ItemSpacing}{0.8ex}
 \setlength{\BasePartSpacing}{2.4ex}
 \setlength{\HotPartSpacing}{1.6ex}

 \newcommand{\plotQueryDist}[2][]{
  #2
   let \p1 = (key min), \p2 = ($(key max |- load max) - (key min)$) in
   plot[samples=100, variable=\t, domain={-0.001:0.471}] (\t * \x2 / 4 + \x1, {(+0.674 * \t*\t*\t -1.539 * \t*\t +1.000 * \t +0.800) * \y2 + \y1}) #1
   plot[samples=100, variable=\t, domain={ 0.469:1.501}] (\t * \x2 / 4 + \x1, {(-4.637 * \t*\t*\t*\t*\t +22.560 * \t*\t*\t*\t -39.607 * \t*\t*\t +30.171 * \t*\t -10.351 * \t +2.318) * \y2 + \y1}) #1
   plot[samples=100, variable=\t, domain={ 1.499:4.001}] (\t * \x2 / 4 + \x1, {0}) #1
  ;
 }
 \begin{minipage}{\linewidth}
  \centering
  \begin{tikzpicture}
   \coordinate (key min);
   \coordinate (key max) at ($(key min) + (\ItemSpacing * 44 + \BasePartSpacing * 4, 0)$);
   \coordinate (load zero) at (key min);
   \coordinate (load max) at ($(load zero) + (0, 3.6ex)$);

   \coordinate (range0 key min) at (key min);
    \coordinate (range0 key max) at ($(key min)!.25!(key max)$);
   \coordinate (range1 key min) at (range0 key max);
    \coordinate (range1 key max) at ($(key min)!.5!(key max)$);
   \coordinate (range2 key min) at (range1 key max);
    \coordinate (range2 key max) at ($(key min)!.75!(key max)$);
   \coordinate (range3 key min) at (range2 key max);
    \coordinate (range3 key max) at (key max);

   \plotQueryDist{\draw[line width=0.4ex]}

   \draw[line width=0.2ex, -latex, line cap=rect]
    (key min) -- ($(key max) + (1.7ex, 0)$) coordinate (tmp);
   \node[inner sep=0.2ex, anchor=south west] (item axis label) at ($(tmp) + (-3.5ex, 1ex)$) {\shortstack[l]{Ordered\\Data}};

   \draw[line width=0.2ex, -latex, line cap=rect]
    (load zero) -- ($(load max) + (0, 1.2ex)$) coordinate (tmp);
   \node[inner sep=0.2ex, anchor=south west] (query axis label) at ($(tmp) + (-1.2ex, 0)$) {Query Density};

   \coordinate (range height) at ($(key min) + (0, -0.5ex)$);

   \coordinate (range0 begin) at ($(range height -| key min) + (\BasePartSpacing * 0.5 - \ItemWidth * 0.5, 0)$);
   \coordinate (range0 end) at ($(range0 begin) + (\ItemSpacing * 11, 0)$);
   \foreach \i [count=\prev from 0] in {1, ..., 3} {
    \coordinate (range\i\space begin) at ($(range\prev\space end) + (\BasePartSpacing, 0)$);
    \coordinate (range\i\space end) at ($(range\i\space begin) + (\ItemSpacing * 11, 0)$);
   }

   \foreach \i in {0, ..., 3} {
    \drawDPU[3.2ex]{($(key min |- range height)!{0.25 * \i + 0.125}!(key max |- range height) + (0, - \ItemHeight - 1.6ex - 1.5ex)$)}{proc\i}
   }
   \foreach \i in {0, ..., 2} {
    \draw[black!70, double, line width=0.2ex] (range\i\space key max |- proc\i.south) -- ($(range\i\space key max |- load max) + (0, 1ex)$);
   }

   \begin{scope}[plain item/.style={shift={#1}, yslant=-0.7, draw=black, fill=black!25, line width=0.15ex}]
    \foreach \i in {3, ..., 0} {
     \foreach \j in {11, ..., 0} {
      \path[plain item={($(range\i\space begin)!{1.0 / 11 * \j}!(range\i\space end)$)}] (0, 0) rectangle (\ItemWidth, -\ItemHeight + \ItemWidth * 0.7);
     }
    }
   \end{scope}
  \end{tikzpicture}
  \subcaption{Data-balanced one-dimensional partitioning.}
  \label{fig:example-partitions:equal-data}
 \end{minipage}

 \begin{minipage}{\linewidth}
  \centering
  \begin{tikzpicture}
   \coordinate (key min);
   \coordinate (key max) at ($(key min) + (\ItemSpacing * 44 + \BasePartSpacing * 4, 0)$);
   \coordinate (load zero) at (key min);
   \coordinate (load max) at ($(load zero) + (0, 3.6ex)$);

   \coordinate (range0 key min) at (key min);
    \coordinate (range0 key max) at ($(key min)!.0625!(key max)$);
   \coordinate (range1 key min) at (range0 key max);
    \coordinate (range1 key max) at ($(key min)!.125!(key max)$);
   \coordinate (range2 key min) at (range1 key max);
    \coordinate (range2 key max) at ($(key min)!.1875!(key max)$);
   \coordinate (range3 key min) at (range2 key max);
    \coordinate (range3 key max) at (key max);

   \plotQueryDist{\draw[line width=0.4ex]}

   \draw[line width=0.2ex, -latex, line cap=rect]
    (key min) -- ($(key max) + (1.7ex, 0)$) coordinate (tmp);
   \node[inner sep=0.2ex, anchor=south west] (item axis label) at ($(tmp) + (-3.5ex, 1ex)$) {\shortstack[l]{Ordered\\Data}};

   \draw[line width=0.2ex, -latex, line cap=rect]
    (load zero) -- ($(load max) + (0, 1.2ex)$) coordinate (tmp);
   \node[inner sep=0.2ex, anchor=south west] (query axis label) at ($(tmp) + (-1.2ex, 0)$) {Query Density};

   \coordinate (range height) at ($(key min) + (0, -0.5ex)$);

   \coordinate (range0 begin) at ($(range height -| key min) + (\BasePartSpacing * 0.5 - \ItemWidth * 0.5, 0)$);
    \coordinate (range0 end) at ($(range0 begin) + (\ItemSpacing * 2, 0)$);
   \coordinate (range1 begin) at ($(range0 end) + (\BasePartSpacing, 0)$);
    \coordinate (range1 end) at ($(range1 begin) + (\ItemSpacing * 2, 0)$);
   \coordinate (range2 begin) at ($(range1 end) + (\BasePartSpacing, 0)$);
    \coordinate (range2 end) at ($(range2 begin) + (\ItemSpacing * 2, 0)$);
   \coordinate (range3 begin) at ($(range2 end) + (\BasePartSpacing, 0)$);
    \coordinate (range3 end) at ($(range3 begin) + (\ItemSpacing * 38, 0)$);

   \path[name path=upper-line] ($(key min) + (0, \ItemHeight * -0.8)$) -- ($(key max) + (0, \ItemHeight * 0.6)$);
   \path[name path=lower-line] ($(key min) + (0, -\ItemHeight - 4.3ex)$) -- ($(key max) + (0, -\ItemHeight - 2.15ex)$);
   \begin{scope}[
    plain item/.style={shift={#1}, yslant=-0.7, draw=black, fill=black!25, line width=0.15ex},
   ]
    \foreach \iDPU/\maxIItem [evaluate=\iDPU as \rightDPU using int(\iDPU + 1)] in {3/38, 2/2, 1/2, 0/2} {
     \ifthenelse{\iDPU = 3}{}{
      \path[name path=key-delim-hpos] (range\iDPU\space key max) -- (range\iDPU\space key max |- proc\rightDPU.south);
      \coordinate (tmp) at ($(key min |- range height)!{0.25 * \iDPU + 0.25}!(key max |- range height)$);
      \path[name path=dpu-delim-hpos] (tmp) -- (tmp |- proc\rightDPU.south);

      \path [name intersections={of=upper-line and key-delim-hpos, by={upper-intsc}}];
      \path [name intersections={of=lower-line and dpu-delim-hpos, by={lower-intsc}}];
      \draw[black!70, double, line width=0.2ex]
       ($(range\iDPU\space key max |- load max) + (0, 1ex)$)
       -- (upper-intsc) -- (lower-intsc)
       -- (tmp |- proc\rightDPU.south);
     }
     \foreach \iItem in {\maxIItem, ..., 0} {
      \path[plain item={($(range\iDPU\space begin)!{1.0 / \maxIItem * \iItem}!(range\iDPU\space end)$)}] (0, 0) rectangle (\ItemWidth, -\ItemHeight + \ItemWidth * 0.7);
     }

     \drawDPU[3.2ex]{($(key min |- range height)!{0.25 * \iDPU + 0.125}!(key max |- range height) + (0, - \ItemHeight - 1.6ex - 2ex)$)}{proc\iDPU}
    }
   \end{scope}
  \end{tikzpicture}
  \subcaption{Query-balanced one-dimensional partitioning.}
  \label{fig:example-partitions:equal-query}
 \end{minipage}

 \begin{minipage}{\linewidth}
  \centering
  \begin{tikzpicture}
   \coordinate (key min);
   \coordinate (key max) at ($(key min) + (\ItemSpacing * 44 + \BasePartSpacing * 4, 0)$);
   \coordinate (load zero) at (key min);
   \coordinate (load max) at ($(load zero) + (0, 3.6ex)$);

   \coordinate (base0 key min) at (key min);
    \coordinate (base0 key max) at ($(key min)!.25!(key max)$);
   \coordinate (base1 key min) at (base0 key max);
    \coordinate (base1 key max) at ($(key min)!.5!(key max)$);
   \coordinate (base2 key min) at (base1 key max);
    \coordinate (base2 key max) at ($(key min)!.75!(key max)$);
   \coordinate (base3 key min) at (base2 key max);
    \coordinate (base3 key max) at (key max);

   \coordinate (cold height) at ($(key min) + (0, -0.5ex)$);
    \coordinate (hot height) at ($(cold height) + (0, -\ItemHeight + \ItemWidth * 0.7 + \ItemSpacing * 0.9)$);

   \coordinate (base0 begin) at ($(cold height -| key min) + (\BasePartSpacing * 0.5 - \ItemWidth * 0.5, 0)$);
    \coordinate (base0 end) at ($(base0 begin) + (\ItemSpacing * 11, 0)$);
   \coordinate (base1 begin) at ($(base0 end) + (\BasePartSpacing, 0)$);
    \coordinate (base1 end) at ($(base1 begin) + (\ItemSpacing * 11, 0)$);
   \coordinate (base2 begin) at ($(base1 end) + (\BasePartSpacing, 0)$);
    \coordinate (base2 end) at ($(base2 begin) + (\ItemSpacing * 11, 0)$);
   \coordinate (base3 begin) at ($(base2 end) + (\BasePartSpacing, 0)$);
    \coordinate (base3 end) at ($(base3 begin) + (\ItemSpacing * 11, 0)$);

   \coordinate (hot0 begin) at ($(hot height -| base0 begin) + (\ItemSpacing * 1, 0)$);
    \coordinate (hot0 end) at ($(hot0 begin) + ({(\ItemSpacing * 8 - \HotPartSpacing * 2) / 6 * 2}, 0)$);
   \coordinate (hot1 begin) at ($(hot0 end) + (\HotPartSpacing, 0)$);
    \coordinate (hot1 end) at ($(hot1 begin) + ({(\ItemSpacing * 8 - \HotPartSpacing * 2) / 6 * 2}, 0)$);
   \coordinate (hot2 begin) at ($(hot1 end) + (\HotPartSpacing, 0)$);
    \coordinate (hot2 end) at ($(hot2 begin) + ({(\ItemSpacing * 8 - \HotPartSpacing * 2) / 6 * 2}, 0)$);

   \begin{scope}
    \clip ($(key min)!.023!(key max)$) rectangle ($(key min |- load max)!.212!(key max |- load max)$);
    \plotQueryDist[coordinate (tmp) -- (tmp |- key min) -| cycle]{\fill[Orange!30]}
   \end{scope}
   \begin{scope}
    \clip ($(key min)!.023!(key max)$) coordinate (tmp) ($(tmp) + (-0.15ex, 0)$) rectangle ($(tmp |- load max) + (0.15ex, 0)$)
          ($(key min)!.086!(key max)$) coordinate (tmp) ($(tmp) + (-0.15ex, 0)$) rectangle ($(tmp |- load max) + (0.15ex, 0)$)
          ($(key min)!.145!(key max)$) coordinate (tmp) ($(tmp) + (-0.15ex, 0)$) rectangle ($(tmp |- load max) + (0.15ex, 0)$)
          ($(key min)!.212!(key max)$) coordinate (tmp) ($(tmp) + (-0.15ex, 0)$) rectangle ($(tmp |- load max) + (0.15ex, 0)$);
    \plotQueryDist[coordinate (tmp) -- (tmp |- key min) -| cycle]{\fill[Orange!70]}
   \end{scope}
   \plotQueryDist{\draw[line width=0.4ex]}

   \draw[line width=0.2ex, -latex, line cap=rect]
    (key min) -- ($(key max) + (1.7ex, 0)$) coordinate (tmp);
   \node[inner sep=0.2ex, anchor=south west] (item axis label) at ($(tmp) + (-3.5ex, 1ex)$) {\shortstack[l]{Ordered\\Data}};

   \draw[line width=0.2ex, -latex, line cap=rect]
    (load zero) -- ($(load max) + (0, 1.2ex)$) coordinate (tmp);
   \node[inner sep=0.2ex, anchor=south west] (query axis label) at ($(tmp) + (-1.2ex, 0)$) {Query Density};

   \coordinate (DPU height) at ($(cold height) + (0, \ItemSpacing * -0.7 -\ItemHeight - 1.6ex - 4ex)$);
   \foreach \i in {1, ..., 3} {
    \drawDPU[3.2ex]{($(key min |- DPU height)!{0.25 * \i + 0.125}!(key max |- DPU height) + (\ItemWidth * 0.5, 0)$)}{proc\i}
    \draw[black!70, double, line width=0.2ex] (base\i\space key min |- proc\i.south) --($ (base\i\space key min |- load max) + (0, 1ex)$);
   }

   \path[name path=hot-migration] ($(hot0 begin) + (0, -\ItemHeight - 2ex)$) -- ($(base3 end) + (0, -\ItemHeight)$);
   \foreach \iHot [evaluate=\iHot as \iDPU using int(\iHot + 1)] in {2, ..., 0} {
    \coordinate (hot from) at ($(hot\iHot\space end) + (-0.25ex, -\ItemHeight + 1ex)$);
    \coordinate (hot to) at (proc\iDPU.north west);
    \path[name path=hot-out] (hot from) -- +(1ex, -4ex);
    \path[name path=out-cont2-hpos] ($(hot from) + (3ex, 1ex)$) -- ++(0, -4ex);
    \path[name path=out-jct-hpos] ($(hot from) + (5ex, 1ex)$) -- ++(0, -4ex);
    \path[name path=in-jct-guide] ($(hot to) + (6ex, 0)$) -- ++(-24ex, 4ex);
    \path[name path=in-cont1-guide] ($(hot to) + (16ex, 0) + (-8.5ex, 1ex)$) -- ++(-25.5ex, 3ex);
    \path[name path=hot-in] (hot to) -- ++(-2ex, 8ex);
    \path[name intersections={of=hot-migration and hot-out, by={out-cont1}}];
    \path[name intersections={of=hot-migration and out-cont2-hpos, by={out-cont2}}];
    \path[name intersections={of=hot-migration and out-jct-hpos, by={out-jct}}];
    \path[name intersections={of=hot-migration and in-jct-guide, by={in-jct}}];
    \path[name intersections={of=hot-migration and in-cont1-guide, by={in-cont1}}];
    \path[name intersections={of=hot-migration and hot-in, by={in-cont2}}];
    \draw[line width=0.8ex, black!60, -{Triangle[width=2ex, length=2ex, bend]},
          shorten >= 0.335ex, preaction={shorten >= 0ex, draw=white, line width=1.1ex, -{Triangle[width=2.485ex, length=2.485ex, bend]}}]
     (hot from) ..controls (out-cont1) and (out-cont2).. (out-jct)
     -- (in-jct) ..controls (in-cont1) and (in-cont2).. (hot to);
   }

   \drawDPU[3.5ex]{($(key min |- DPU height)!{0.125}!(key max |- DPU height) + (\ItemWidth * 0.5, 0)$)}{proc0}

   \begin{scope}[
    plain item/.style={shift={#1}, yslant=-0.7, draw=black, fill=black!25, line width=0.15ex},
   ]
    \foreach \i in {11, ..., 0} {
     \path[plain item={($(base3 begin)!{1.0 / 11 * \i}!(base3 end)$)}] (0, 0) rectangle (\ItemWidth, -\ItemHeight + \ItemWidth * 0.7);
    }

    \foreach \i in {11, ..., 0} {
     \path[plain item={($(base2 begin)!{1.0 / 11 * \i}!(base2 end)$)}] (0, 0) rectangle (\ItemWidth, -\ItemHeight + \ItemWidth * 0.7);
    }

    \foreach \i in {11, ..., 0} {
     \path[plain item={($(base1 begin)!{1.0 / 11 * \i}!(base1 end)$)}] (0, 0) rectangle (\ItemWidth, -\ItemHeight + \ItemWidth * 0.7);
    }

    \foreach \i in {11, ..., 10} {
     \path[plain item={($(base0 begin)!{1.0 / 11 * \i}!(base0 end)$)}] (0, 0) rectangle (\ItemWidth, -\ItemHeight + \ItemWidth * 0.7);
    }
    \foreach \i in {9, ..., 1} {
     \path[opacity=0.25, plain item={($(base0 begin)!{1.0 / 11 * \i}!(base0 end)$)}] (0, 0) rectangle (\ItemWidth, -\ItemHeight + \ItemWidth * 0.7);
    }
    \foreach \i in {2, ..., 0} {
     \path[plain item={($(hot2 begin)!{1.0 / 2 * \i}!(hot2 end)$)}, draw=Orange!70!black] (0, 0) rectangle (\ItemWidth, -\ItemHeight + \ItemWidth * 0.7);
    }
    \foreach \i in {2, ..., 0} {
     \path[plain item={($(hot1 begin)!{1.0 / 2 * \i}!(hot1 end)$)}, draw=Orange!70!black] (0, 0) rectangle (\ItemWidth, -\ItemHeight + \ItemWidth * 0.7);
    }
    \foreach \i in {2, ..., 0} {
     \path[plain item={($(hot0 begin)!{1.0 / 2 * \i}!(hot0 end)$)}, draw=Orange!70!black] (0, 0) rectangle (\ItemWidth, -\ItemHeight + \ItemWidth * 0.7);
    }
    \foreach \i in {0} {
     \path[plain item={($(base0 begin)!{1.0 / 11 * \i}!(base0 end)$)}] (0, 0) rectangle (\ItemWidth, -\ItemHeight + \ItemWidth * 0.7);
    }
   \end{scope}
  \end{tikzpicture}
  \subcaption{Query density-driven partitioning.}
  \label{fig:example-partitions:hot_cold}
 \end{minipage}
 \caption{Query density-driven partitioning (\subref{fig:example-partitions:hot_cold}) versus one-dimensional partitioning schemes (\subref{fig:example-partitions:equal-data}, \subref{fig:example-partitions:equal-query}) for skewed queries.}
 \label{fig:example-partitions}
\end{figure}

\subsection{Overview of Partitioning}
\label{sec:load_balance:overview}

On the basis of the key ideas aforementioned, we design our
query density-driven offline key-range partitioning scheme.  Its
entire algorithm is summarized in \cref{algo:hot-cold}, where the
notations used in algorithmic descriptions are summarized in
\cref{tab:notation}.  As seen from \cref{algo:hot-cold}, our partitioning
algorithm assumes two things.  One is that input key-value pairs are
sorted in key order and uniformly divided into data chunks $c_i$, where
the chunk size is a parameter to tune spatial locality.  The size
information is given as \SizeOf at the algorithmic level.  The other is
that a reference workload is given as \NQrys, which is supposed to count
the number of hit queries by converting range queries to point ones.

\cref{algo:hot-cold} proceeds as follows.  First, it divides the key
space into the same number of ranges as DPUs $P$ so that the data size
is evenly distributed, forming \emph{base partitions}
(\lineref{algo:hot-cold:base}).  If the workload is not so skewed and
all their ranges are cold, the partitioning process is complete.  If hot
ranges inhabit base partitions, the imbalance of query loads will arise,
which shall be addressed in the following steps.  Second, it identifies
hot ranges within each base partition and splits hot partitions out of
it, which is modeled as \FindHot in \lineref{algo:hot-cold:hot}.  After
removing the hot partitions, the remaining ranges should consist only of
cold ranges with sufficiently low query counts.  These are now cold
partitions (\lineref{algo:hot-cold:cold}).  Lastly, it distributes the
hot partitions.  A DPU that has a relatively lower query load on its
cold partitions receives at most one hot partition
(\lineref{algo:hot-cold:pairing:begin}--\ref{algo:hot-cold:pairing:end}). To
realize this assignment, \FindHot has to ensure that the number of
extracted hot ranges is at most equal to the number of DPUs.

By design, the number of partitions is limited to at most $3P$ because
at most $P$ hot partitions are extracted from $P$ base partitions.  It
will bring smaller memory footprints for query routing, suppressing
CPU cache misses.

In the following subsections, we specify \FindHot, which is the
technical heart of query density-driven partitioning.

\begin{table}[bt]
 \caption{Symbols used for algorithmic descriptions.}
 \label{tab:notation}
 \centering
 \settowidth{\notationsymwidth}{$\NQrys(S)$}
 \begin{tabular}{lp{\dimexpr\linewidth-\notationsymwidth-4\tabcolsep\relax}} \toprule
  Symbol & Definition \\ \midrule
  $\{ c_i \}_{i=1}^{N}$ & $N$ data chunks $c_i$ in key order. \\
  $\NQrys(S)$ & Function that takes a set of data chunks $S$ and returns the number of queries hitting $S$. \\
  $\SizeOf(S)$ & Function that takes a set of data chunks $S$ and returns the total size of data in $S$. \\
  $P$ & Number of DPUs. \\
  $\mathcal{N}_P$ & A set of DPU numbers $1, \ldots, P$.\\
  $Q$ & Number of queries in the reference workload, i.e., $\NQrys\left(\{ c_i \}_{i=1}^{N}\right)$. \\
  $D$ & Total size of data, i.e., $\SizeOf\left(\{ c_i \}_{i=1}^{N}\right)$. \\
  \bottomrule
 \end{tabular}
\end{table}

\begin{algorithm}[bt]
 \caption{Construction of hot/cold partitions.}
 \label{algo:hot-cold}
 \Input{Data chunks $E = \{ c_i \}_{i=1}^{N}$}
 \Output{Routing table $T\colon \mathcal{N}_P \to \mathcal{P}(E)$}
 \Function{$\HotColdPartition(E)$} {
  $B \gets \left\{\left\{c_j\right\}_{j=s_i}^{e_i}\right\}_{i=1}^{P}$, with $\SizeOf\left(\{c_j\}_{j=s_i}^{e_i}\right)$ balanced\; \label{algo:hot-cold:base}
  $H \gets \bigcup\{\FindHot(b) \mid b \in B\}$\;  \label{algo:hot-cold:hot}
  $C \gets \left\{d \mapsto b \setminus \bigcup H \relmiddle| d \in \mathcal{N}_P \land b = \left\{c_j\right\}_{j=s_d}^{e_d}\right\}$\;  \label{algo:hot-cold:cold}
  $T \gets \emptyset$\;
  \For {$d \mapsto c \in C$ in ascending order of $\NQrys(c)$} {  \label{algo:hot-cold:pairing:begin}
    $h \gets \argmax_{h' \in H}\NQrys(h')$ \KwSty{if} $H \ne \emptyset$ \KwSty{otherwise} $\emptyset$\;
    $H \gets H \setminus h$\;
    $T \gets T \cup \{d \mapsto h \cup c\}$\; \label{algo:hot-cold:pairing:end}
  }
  \Return{$T$}\;
 }
\end{algorithm}

\subsection{Greedy Hot Range Selection}
\label{sec:load_balance:greedy}

\begin{algorithm}[bt]
 \caption{Greedy hot range selection.}
 \label{algo:hot_select}
 \Input{A base partition $b = \{ c_i \}_{i=s}^e$}
 \Output{A set of hot ranges $h$ from $b$}
 \Function{$\FindHot(b)$} {
  $(h, l) \gets (\emptyset, s)$\;
  \For{$r \gets s$ \KwSty{to} $e$}{
   $l' \gets \max \left\{ l' \in (l,r] \relmiddle| \SizeOf\left(\{ c_i \}_{i=l'}^r\right) \ge \frac{1}{\alpha}\frac{D}{P} \right\}$\;  \label{algo:hot_select:window}
   $l \gets \max \{l, l'\}$\;  \label{algo:hot_select:l}
   \If{$\NQrys\left(\{ c_i \}_{i=l}^r\right) \ge \frac{Q}{P}$  \label{algo:hot_select:thres}} {
    $h \gets h \cup \left\{\{ c_i \}_{i=l}^r\right\}$\;  \label{algo:hot_select:hot}
    $l \gets r + 1$\;  \label{algo:hot_select:next_scan}
   }
  }
  \Return{$h$}\;
 }
\end{algorithm}

As mentioned in \cref{sec:load_balance:idea}, we should identify hot
ranges of high density.  Now, the window width used to calculate density
is a crucial parameter that should be adjusted for a trade-off,
depending on the given workload.  Narrow windows bring small hot partitions,
resulting in data size balance, whereas wide
windows bring more hot partitions to move, resulting in query load
balance.  To allow us to select a sweet spot in the trade-off, we
introduce a parameter $\alpha$, the ratio of the data size of a base
partition to the maximum width of a hot partition.  It is a positive
integer, with a larger value indicating a greater emphasis on data size
balancing and a value closer to 1 indicating a greater emphasis on query
load balancing.
We consider the optimal $\alpha$ to be the minimum $\alpha$ that allows
the data size to fit in memory.

\cref{algo:hot_select} shows a simple greedy algorithm based on such
$\alpha$ for \FindHot.  Note that $\max S$ produces $-\infty$ if $S =
\emptyset$ throughout this paper.
This algorithm scans the base partition by the
windows of size $\frac{1}{\alpha}\frac{D}{P}$
(\lineref{algo:hot_select:window}) and extracts data chunks as hot
partitions if the number of queries targeting those data chunks exceeds
the threshold $\frac{Q}{P}$, where $Q$ is the total number of queries in
the reference workload (\lineref{algo:hot_select:thres}).

This greedy hot range selection algorithm conforms to the basic design
of query density-driven partitioning because the number of hot ranges is at most
$P$. Besides, it has the following properties:
\begin{itemize}
 \item The data size of each hot partition is limited to
       $\frac{1}{\alpha}\frac{D}{P} + M_d$, where $M_d$ is the maximum
       size of data per chunk, i.e., $M_d = \max_{1 \le i \le N}
       \SizeOf\left( \{c_i\} \right)$;
 \item The query load on each hot partition is limited to $\frac{Q}{P} +
       M_q$, where $M_q$ is the maximum number of queries targeting a
       single data chunk, i.e., $M_q = \max_{1 \le i \le N} \NQrys\left(
       \{c_i\} \right)$;
 \item The query load for cold partitions remaining in each DPU is
       suppressed below $\alpha\frac{Q}{P}$.
\end{itemize}
The query load on each DPU is thus controlled to be less than
$(\alpha+1)\frac{Q}{P} + M_q$, and the data size is done to be
less than $\left(\frac{1}{\alpha} + 1\right)\frac{D}{P} + M_d$.
\ifdefined\printProof
The proof is in \cref{sec:proof:greedy}.
\else
The proof shall be in an extended version of this
paper.\footnote{The extended version is available at
\url{https://doi.org/10.48550/arXiv.XXXX.XXXXX}.}%
\setcounter{extfootnote}{\value{footnote}}%
\fi

\subsection{Double-Scan Hot Range Selection}
\label{sec:load_balance:double}

\begin{figure}[bt]
 \setlength{\ItemWidth}{1.3ex}
 \setlength{\ItemHeight}{3.3ex}
 \setlength{\ItemSpacing}{0.8ex}
 \setlength{\BasePartSpacing}{2.4ex}
 \setlength{\HotPartSpacing}{1.6ex}

 \centering
 \begin{tikzpicture}
  \coordinate (key min);
  \coordinate (key max) at ($(key min) + (\ItemSpacing * 44 + \BasePartSpacing * 4, 0)$);
  \coordinate (load zero) at (key min);
  \coordinate (load max) at ($(load zero) + (0, 4ex)$);

  \coordinate (range0 key min) at (key min);
   \coordinate (range0 key max) at ($(key min)!.25!(key max)$);
  \coordinate (range1 key min) at (range0 key max);
   \coordinate (range1 key max) at ($(key min)!.5!(key max)$);
  \coordinate (range2 key min) at (range1 key max);
   \coordinate (range2 key max) at ($(key min)!.75!(key max)$);
  \coordinate (range3 key min) at (range2 key max);
   \coordinate (range3 key max) at (key max);

  \draw[line width=0.4ex]
   let \p1 = (key min), \p2 = ($(key max |- load max) - (key min)$) in
   (\x1, \y2 + \y1) -- (0.25 * \x2 + \x1, \y2 + \y1)
   -- (0.3 * \x2 + \x1, 0.25 * \y2 + \y1)
   -- (\x2 + \x1, 0.25 * \y2 + \y1);

  \node[inner sep=0.2ex, anchor=east] (thres density) at ($(load max) + (0, 0.7ex)$) {$\alpha\frac{Q}{D}$};
  \draw[line width=0.2ex, dashed] (thres density.east) -- (thres density.center -| range0 key max);

  \coordinate (window min) at ($(range0 key min)!.517!(range0 key max)$);
  \coordinate (window max) at ($(range0 key min)!.85!(range0 key max)$);
  \path[pattern color=Blue!60, pattern={Lines[angle=45, line width=0.15em, distance=0.3em]}] (window min) rectangle (window max |- load max);
  \draw[Blue, line width=0.4ex, {Straight Barb[right, length=0.7ex]}-{Straight Barb[left, length=0.7ex]}] ($(window min |- thres density.center) + (0, 1.6ex)$) -- ($(window max |- thres density.center) + (0, 1.6ex)$)
   coordinate[midway] (tmp);
  \node[text=Blue!70!black, inner sep=0.2ex, anchor=south] (window label) at ($(tmp) + (0, 1ex)$) {$\frac{1}{\alpha}\frac{D}{P}$};
  \draw[Blue!40, line width=0.6ex, preaction={draw=Blue, line width=0.8ex}]
   (window min) -- ($(window min |- thres density.center) + (0, 2ex)$)
   (window max) -- ($(window max |- thres density.center) + (0, 2ex)$);
  \draw[Blue, {Circle[length=1ex]}-, line width=0.5ex, preaction={draw=white, {Circle[length=1.4ex]}-, shorten <= -0.2ex, line width=0.9ex}]
   ($(window min |- thres density.center)!.5!(window max |- thres density.center) + (0, -3ex)$)
   -- ($(range0 key max |- load max) + (0.5ex, 3.3ex)$) -- ++(2.5ex, 0)
   coordinate (message anchor);

  \draw[line width=0.2ex, -latex, line cap=rect]
   (key min) -- ($(key max) + (1.7ex, 0)$) coordinate (tmp);
  \node[inner sep=0.2ex, anchor=south west] (item axis label) at ($(tmp) + (-3.5ex, 1.3ex)$) {\shortstack[l]{Ordered\\Data}};

  \draw[line width=0.2ex, -latex, line cap=rect]
   (load zero) -- ($(load max) + (0, 3.5ex)$) coordinate (tmp);
  \node[inner sep=0.2ex, anchor=south west] (query axis label) at ($(tmp -| thres density.west) + (0, 0)$) {\shortstack[l]{Query\\Density}};

  \coordinate (range height) at ($(key min) + (0, -0.5ex)$);

  \coordinate (range0 begin) at ($(range height -| key min) + (\BasePartSpacing * 0.5 - \ItemWidth * 0.5, 0)$);
  \coordinate (range0 end) at ($(range0 begin) + (\ItemSpacing * 11, 0)$);
  \foreach \i [count=\prev from 0] in {1, ..., 3} {
   \coordinate (range\i\space begin) at ($(range\prev\space end) + (\BasePartSpacing, 0)$);
   \coordinate (range\i\space end) at ($(range\i\space begin) + (\ItemSpacing * 11, 0)$);
  }

  \foreach \i in {0, ..., 3} {
   \drawDPU[3.2ex]{($(key min |- range height)!{0.25 * \i + 0.125}!(key max |- range height) + (0, - \ItemHeight - 1.6ex - 1.5ex)$)}{proc\i}
  }
  \foreach \i in {0, ..., 2} {
   \draw[black!70, double, line width=0.2ex] (range\i\space key max |- proc\i.south) -- ($(range\i\space key max |- thres density.center) + (0, 1ex)$);
  }

  \begin{scope}[plain item/.style={shift={#1}, yslant=-0.7, draw=black, fill=black!25, line width=0.15ex}]
   \foreach \i in {3, ..., 0} {
    \foreach \j in {11, ..., 0} {
     \path[plain item={($(range\i\space begin)!{1.0 / 11 * \j}!(range\i\space end)$)}] (0, 0) rectangle (\ItemWidth, -\ItemHeight + \ItemWidth * 0.7);
    }
   }
  \end{scope}

  \node[draw=Blue, text=Blue, fill=white, line width=0.35ex, font=\bfseries, anchor=south west] at ($(message anchor) + (0, -3ex)$) {\shortstack[l]{
   \#Queries $\bm < \frac{Q}{P}$ \\
   $\bm\Rightarrow$ No hot ranges
  }};
 \end{tikzpicture}
 \caption{Too few hot ranges selected in \cref{algo:hot_select}: no candidate window contains more queries than the threshold, and the query load becomes imbalanced.}
 \label{fig:too_small_nr_hot}
\end{figure}
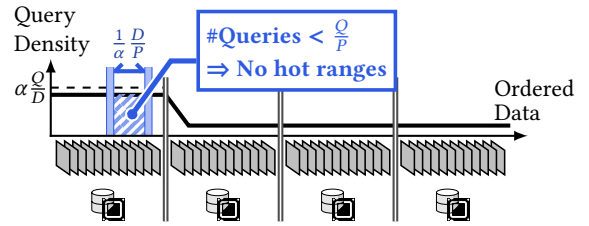

Unfortunately, \cref{algo:hot_select} fails to prevent the cases of an
insufficient number of hot partitions found, even though the query load
is unevenly distributed among the base partitions.  This shortcoming is
particularly apparent when $\alpha$ is large, as illustrated in
\cref{fig:too_small_nr_hot}.  As a result, some DPUs cannot receive hot
partitions and will be nearly idle after the partitioning.

To address this problem, we present double-scan extension, which
introduces the second scan that selects hot ranges with more relaxed
criteria if necessary. The number of hot partitions that the second scan
tries to select is determined by the number of queries remaining in each
base partition after \cref{algo:hot_select} runs.

\begin{algorithm}[bt]
 \caption{Double-scan hot range selection.}
 \label{algo:double-scan}
 \Input{A base partition $b = \{ c_i \}_{i=s}^e$}
 \Output{A set of hot ranges $h$ from $b$}
 \Function{$\FindHot(b)$} {
  $h \gets \FindHotGreedy(b)$ \tcp*{Alias of \cref{algo:hot_select}} \label{algo:double-scan:greedy}
  $h^* \gets \bigcup h$\;
  $\beta \gets \left\lfloor \NQrys(b \setminus h^*) / \frac{Q}{P} \right\rfloor$\;  \label{algo:double-scan:beta}
  \lIf(\label{algo:double-scan:early}){$\beta = 0$}{\Return{$h$}}
  $(l, l_0, r_0) \gets (s, s, s)$\;  \label{algo:double-scan:init_max}
  \For{$r \gets s$ \KwSty{to} $e$ \label{algo:double-scan:for}} {
   $l' \gets \max\left\{ l' \in (l,r] \relmiddle| \SizeOf\left(\{ c_i \}_{i=l'}^r\right) \ge \frac{\beta}{\alpha}\frac{D}{P} \right\}$\;
   $l \gets \max\{l, l'\}$\;  \label{algo:double-scan:window}
   \lIf{$\NQrys\left(\{ c_i \}_{i=l}^r \setminus h^* \right) > \NQrys\left(\{ c_i \}_{i=l_0}^{r_0} \setminus h^* \right)$  \label{algo:double-scan:cmp_nqry}} {
    $(l_0, r_0) \gets (l, r)$  \label{algo:double-scan:update_max}
   }
  }
  $r \gets r_0$\;  \label{algo:double-scan:init_r}
  \While{$l_0 \le r$  \label{algo:double-scan:while}} {
   $l \gets \max\left\{ l' \in (l_0,r] \relmiddle| \SizeOf\left(\{ c_i \}_{i=l'}^r\right) \ge \frac{1}{\alpha}\frac{D}{P} \right\}$\;  \label{algo:double-scan:divide}
   $l \gets \max\{l_0, l\}$\;  \label{algo:double-scan:small_window}
   $h \gets h \cup \left(\left\{\{c_i\}_{i=l}^r \setminus h^* \right\} \setminus \{\emptyset\}\right)$\;  \label{algo:double-scan:append}
   $r \gets l - 1$\;  \label{algo:double-scan:next_r}
  }
  \Return{$h$}\;
 }
\end{algorithm}

\Cref{algo:double-scan} shows hot range selection with the double-scan
sophistication.  Let $q_{left}$
be the number of queries left in the base partition and $\beta \coloneqq
\left\lfloor \frac{q_{left}}{Q/P} \right\rfloor$
(\lineref{algo:double-scan:beta}).  It tries to find $\beta$ hot ranges.
This design is intended to select a hot range for approximately every
$\frac{Q}{P}$ queries.  It scans the base partition in windows of data
size $\frac{\beta}{\alpha}\frac{D}{P}$
(\lineref{algo:double-scan:window}) and finds the location of the window
with the maximum number of queries left in it
(\lineref{algo:double-scan:cmp_nqry}).  It then divides the data chunks
overlapping the window at that location into $\beta$ equal parts based
on data size (\lineref{algo:double-scan:divide}), removes the chunks
in existing hot ranges, and selects the resulting parts as new
hot ranges (\lineref{algo:double-scan:append}).

This double-scan algorithm satisfies all properties achieved by the
greedy algorithm in \cref{sec:load_balance:greedy}.  Furthermore, the
upper limit of the query load on cold partitions remaining in each DPU
is reduced to $\frac{\alpha+1}{3}\frac{Q}{P}$ for $\alpha > 1$.  For example,
when $\alpha$ is 10, this reduces from $10\frac{Q}{P}$ to
$\frac{11}{3}\frac{Q}{P}$, which is roughly a 63\%
reduction.  As a result, the query load on each DPU is controlled to be
less than $\frac{\alpha+4}{3}\frac{Q}{P} +
M_q$, while keeping the data size less than $\left(\frac{1}{\alpha} +
1\right)\frac{D}{P} + M_d$.
\ifdefined\printProof
The proof is in \cref{sec:proof:double}.
\else
The proof shall be in an extended version of this
paper.\footnotemark[\value{extfootnote}]
\fi

\subsection{Computational Cost}

The time complexity of \cref{algo:hot-cold} is $O(P)$, and those of
\cref{algo:hot_select,algo:double-scan} are $O(N)$.  Hence, our
partitioning algorithm runs in $O(P + N)$ time.  This cost is negligible
in practice to the extent of our experience.  In the worst-case full
rebalancing scenario, for example, the cost of partitioning is less by
three orders of magnitude than that of the other parts: the
(de)serialization and data transfer of key-value pairs.

\section{\BpForest}
\label{sec:bp-forest}
We implement \BpForest, a key-value store that adopts query
density-driven partitioning, dedicated to batched RAQs with commutative
reduction operators.

\subsection{Data Representation}
\label{sec:bp-forest:data-repr}
\BpForest performs partitioning as the initialization as follows.  It
first distributes the key-value pairs evenly across the DPUs to form
base partitions.  Then, it constructs the routing table on the CPU from
a given reference workload by using the proposed algorithm, treating
subtrees of height 1 of \BpTrees in DPUs as data chunks.  Lastly, the
key-value pairs of hot partitions are exchanged among DPUs.  Rebalancing
in operation is left for future work.

Each DPU maintains two \BpTrees for the data in its assigned partitions;
one is for the cold partitions and the other is for the hot partition.
We employ the ``occupancy embedding'' technique~\cite{hideshima25:b+tree_UPMEM},
where both key and value are embedded in the node of 256 bytes
without boxing.

The routing table is represented by two arrays: the partition boundary
array and the destination DPU array, as illustrated in
\cref{fig:routing-table-impl}.  The DPUs responsible for any RAQ are
identified through binary search followed by linear search on the
partition boundary array.  Since $3P$ partitions at most are
constructed, the routing table is sufficiently small to fit into the CPU
cache.

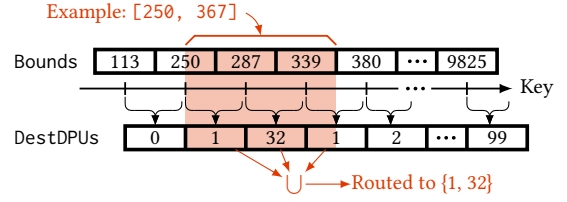
\begin{figure}[tb]
 \centering
 \small
 \begin{tikzpicture}[
  table cell/.style={inner ysep=0.5ex, inner xsep=0ex, align=center, text width=6.5ex},
  table grid/.style={line width=0.4ex},
  axis line/.style={line width=0.2ex}
 ]
  \node[table cell] (bound0) {113};
  \foreach \boundary [count=\i from 1, count=\prev from 0] in {250, 287, 339, 380} {
   \node[table cell, anchor=west] (bound\i) at (bound\prev.east) {\boundary};
  }
  \draw[ellipsis=0.4ex] ($(bound4.east) + (1ex, 0)$) -- ++(2.4ex, 0) coordinate (tmp);
  \node[table cell, anchor=west] (bound6) at ($(tmp) + (1ex, 0)$) {9825};
  \node[anchor=east, font=\ttfamily] (bounds label) at ($(bound0.west) + (-1ex, 0)$) {Bounds};

  \draw[table grid] (bound0.north west) rectangle (bound6.south east);
  \foreach \i in {0, ..., 4} {
   \draw[table grid] (bound\i.north east) -- (bound\i.south east);
  }
  \draw[table grid] (bound6.north west) -- (bound6.south west);

  \coordinate (axis height) at ($(bound0.south) + (0, -1.7ex)$);
  \foreach \i in {0, ..., 4, 6} {
   \draw[axis line] ($(bound\i.center |- axis height) + (0, 0.7ex)$) -- ++(0, -1.4ex) coordinate (scale mark \i);
  }
  \draw[axis line] ($(bound0.center |- axis height) + (-5ex, 0)$)
   -- ($(bound4.center |- axis height) + (3.5ex, 0)$) coordinate (tmp);
  \coordinate (scale mark 5) at ($(scale mark 0 -| tmp) + (1.8ex, 0)$);
  \draw[ellipsis=0.4ex] ($(tmp) + (0.6ex, 0)$) -- ++(2.4ex, 0) coordinate (tmp);
  \draw[axis line, -latex] ($(tmp) + (0.6ex, 0)$)
   -- ($(bound6.center |- axis height) + (5ex, 0)$)
   coordinate (axis top)
   node[anchor=west] {Key};
  \coordinate (scale mark 7) at (axis top |- scale mark 6);

  \coordinate (dest height) at ($(axis height) + (0, -5.4ex)$);
  \foreach \i/\iDPU in {0/0, 1/1, 2/32, 3/1, 4/2, 6/99} {
   \node[table cell, anchor=center] (dest\i) at (bound\i.east |- dest height) {\iDPU};
  }
  \draw[ellipsis=0.4ex] ($(dest4.east) + (3.4ex, 0)$) -- ($(dest6.west) + (-3.4ex, 0)$);
  \node[anchor=west, font=\ttfamily] (dest label) at (bounds label.west |- dest height) {DestDPUs};

  \draw[table grid] (dest0.north west) rectangle (dest6.south east);
  \foreach \i in {0, ..., 4} {
   \draw[table grid] (dest\i.north east) -- (dest\i.south east);
  }
  \draw[table grid] (dest6.north west) -- (dest6.south west);

  \foreach \i [evaluate=\i as \j using int(\i + 1)] in {0, ..., 4, 6} {
   \coordinate (mid scale marks) at ($(scale mark \i)!.5!(scale mark \j)$);
   \draw ($(scale mark \i) + (0, -0.2ex)$) -- ++(0, -0.5ex)
    ..controls +(0, -0.5ex) and +(-0.2ex, 0.2ex).. +(0.2ex, -0.8ex)
    ..controls +(0.2ex, -0.2ex) and +(-0.5ex, 0).. +(1ex, -1ex) coordinate (tmp);
   \draw[-latex] (tmp) -- ($(tmp -| mid scale marks) + (-1ex, 0)$)
    ..controls +(0.5ex, 0) and +(-0.2ex, 0.2ex).. +(0.8ex, -0.2ex)
    ..controls +(0.2ex, -0.2ex) and +(0, 0.5ex).. +(1ex, -1ex)
    -- (mid scale marks |- dest\i.north);

   \draw ($(scale mark \j) + (-0.4ex, -0.7ex)$)
    ..controls +(0, -0.5ex) and +(0.2ex, 0.2ex).. +(-0.2ex, -0.8ex)
    ..controls +(-0.2ex, -0.2ex) and +(0.5ex, 0).. +(-1ex, -1ex) coordinate (tmp);
   \draw[-latex] (tmp) -- ($(tmp -| mid scale marks) + (1ex, 0)$)
    ..controls +(-0.5ex, 0) and +(0.2ex, 0.2ex).. +(-0.8ex, -0.2ex)
    ..controls +(-0.2ex, -0.2ex) and +(0, 0.5ex).. +(-1ex, -1ex)
    -- (mid scale marks |- dest\i.north);
  }

  \node[Orange, anchor=east] (qry) at ($(bound2.north) + (0, 4ex)$) {Example: \texttt{[250, 367]}};
  \draw[Orange, line width=0.2ex] ($(bound1.north) + (0, 0.5ex)$) -- ++(1ex, 1ex) -- ($(bound3.north east) + (-1ex, 1.5ex)$)
   coordinate[midway] (tmp)
   -- ++(1ex, -1ex);
  \draw[Orange, -latex, bend left] (qry.east) to (tmp);
  \begin{scope}[on background layer]
   \fill[Orange, opacity=0.3] (bound1.north) rectangle (dest3.south);
  \end{scope}
  \node[Orange, inner sep=0.2ex, anchor=north] (bigcup) at ($(dest2.south) + (2ex, -2.5ex)$) {$\bigcup$};
  \draw[Orange, -latex] (dest1) -- (bigcup);
  \draw[Orange, -latex] (dest2) -- (bigcup);
  \draw[Orange, -latex] (dest3) -- (bigcup);
  \node[Orange, inner xsep=0ex, inner ysep=0.2ex, anchor=west] (result) at ($(bigcup.east) + (5ex, 0)$) {Routed to \{1, 32\}};
  \draw[Orange, -latex] (bigcup) -- (result);
 \end{tikzpicture}
 \caption{Implementation of the routing table with an example of routing.}
 \label{fig:routing-table-impl}
\end{figure}

\begin{figure*}[tb]
  \begin{tikzpicture}[
  raq/.style={shift={#1}, inner sep=0pt, minimum width=\ItemWidth, minimum height={\ItemHeight - \ItemWidth * 0.7}, yslant=-0.7, anchor=north west,
              draw=black, fill=Blue!20, text=Blue, font=\footnotesize, line width=0.15ex},
  answer/.style={raq={#1}, fill=red!20, text=red},
  procedure/.style={draw, line width=0.7ex, fill=white, inner ysep=1.4ex, inner xsep=1.4ex},
  proc num/.style={circle, fill, text=white, font=\bfseries, inner sep=0.3ex},
  data flow/.style={#1, rounded corners=1.8ex,
                    {Triangle Cap[cap angle=105, reversed]}-{Triangle Cap[cap angle=105]}, line width=1.6ex,
                    shorten <= 0.406ex, shorten >= 0.252ex, preaction={draw=white, line width=2ex, shorten <= 0pt, shorten >= 0pt}},
 ]
  \setlength{\ItemWidth}{2ex}
  \setlength{\ItemHeight}{5ex}
  \setlength{\ItemSpacing}{1ex}
  \coordinate (batch anchor);
  \foreach \i in {8, ..., 0} {
   \node[raq={($(batch anchor) + (\ItemSpacing * \i, 0)$)}] (raq \i) at (0, 0) {Q};
  }
  \node[bbox={(raq 0) (raq 8)}] (batch) {};
  \node[anchor=north] at ($(batch.south) + (0, -0.5ex)$) {A batch of RAQs};

  \node[procedure, anchor=west] (routing label) at ($(batch.east) + (16.3ex, 0)$) {
   \hspace*{0.2ex}
   \shortstack[l]{
    \bfseries Routing
   }
  };
  \node[proc num, anchor=north] (step 1) at (routing label.north west) {1};

  \coordinate (table anchor) at ($(routing label.west) + (-9.2ex, -4.9ex)$);
  \begin{scope}[
   cell/.style={inner sep=0.3ex, font=\footnotesize},
   header cell/.style={cell, font=\footnotesize\bfseries, minimum height=2.5ex},
  ]
   \node[header cell, anchor=north west] (table key header) at (table anchor) {key range};
   \node[cell, anchor=north west] (table key 0) at (table key header.south west) {113--249};
   \node[cell, anchor=north west] (table key 1) at (table key 0.south west) {250--286};
   \node[cell, anchor=north west] (table key 2) at (table key 1.south west) {287--338};
   \node[bbox={(table key header) (table key 2)}] (table col key) {};

   \node[header cell, anchor=north west] (table dpu header) at ($(table col key.north east) + (0.5ex, 0)$) {DPU};
   \node[cell, anchor=north east] (table dpu 0) at (table dpu header.south east) {0};
   \node[cell, anchor=north east] (table dpu 1) at (table dpu 0.south east) {1};
   \node[cell, anchor=north east] (table dpu 2) at (table dpu 1.south east) {32};
   \node[bbox={(table dpu header) (table dpu 2)}] (table col dpu) {};
  \end{scope}
  \begin{scope}[on behind layer]
   \fill[black!20] (table key header.south west) rectangle (table dpu header.north east);
   \fill[white] (table col key.south west) rectangle (table dpu header.south east);
  \end{scope}
  \node[bbox={(table col key) (table col dpu)}] (table cells) {};
  \draw[ellipsis=0.3ex] ($(table cells.south) + (0, 0.5ex)$) -- ++(0, -1.8ex) coordinate (tmp);
  \node[bbox={(table cells) (tmp)}] (routing table) {};
  \node[anchor=north] (routing table label) at ($(routing table.south) + (0, 0.2ex)$) {Routing table};

  \setlength{\ItemWidth}{2.8ex}
  \setlength{\ItemHeight}{3.9ex}
  \setlength{\ItemSpacing}{0.8ex}
  \coordinate (raq ids anchor) at ($(routing label.east) + (2.8ex, -7.3ex)$);
  \begin{scope}[
   raq id/.style={shift={#1}, inner sep=0pt, minimum width=\ItemWidth, minimum height={\ItemHeight - \ItemWidth * 0.3}, yslant=-0.3, anchor=north west,
                  draw=black, fill=white, font=\scriptsize\ttfamily, line width=0.08ex},
  ]
   \foreach \iRow/\RaqID in {0/5, 1/9, 2/3} {
    \foreach \i in {3, ..., 0} {
     \node[raq id={($(raq ids anchor) + (\ItemSpacing * \i, 0) + (0.3ex * \iRow, -1.8ex * \iRow)$)}] (raq id \iRow-\i) at (0, 0) {{[\RaqID]}};
    }
   }
  \end{scope}
  \node[bbox={(raq id 0-0) (raq id 2-3)}] (raq ids) {};
  \node[anchor=north] at ($(raq ids.south) + (0, -0.5ex)$) (raq ids label) {IDs of RAQs};

  \node[procedure, anchor=east] (postprocessing label) at ($(raq ids label.south west) + (-0.3ex, -3.3ex)$) {
   \hspace*{0.2ex}
   \shortstack[l]{
    \bfseries Postprocessing
   }
  };
  \node[proc num, anchor=north] (step 3) at (postprocessing label.north west) {3};

  \coordinate (subbatch anchor) at ($(batch anchor -| routing label.east) + (18ex, 0)$);
  \setlength{\ItemWidth}{2ex}
  \setlength{\ItemHeight}{5ex}
  \setlength{\ItemSpacing}{1ex}
  \foreach \iRow in {0, ..., 2} {
   \foreach \i in {3, ..., 0} {
    \node[raq={($(subbatch anchor) + (\ItemSpacing * \i, 0) + (0, -\ItemHeight * \iRow) + (0.4ex * \iRow, 1.15ex * \iRow)$)}]
     (routed \iRow-\i) at (0, 0) {Q};
   }
   \node[bbox={(routed \iRow-0) (routed \iRow-3)}] (subbatch \iRow) {};
  }
  \node[bbox={(subbatch 0) (subbatch 2)}] (subbatch) {};
  \node[anchor=north] (subbatch label) at ($(subbatch.south) + (0, -0.5ex)$) {\shortstack[l]{
   Sub-batches\\
   for each DPU
  }};

  \setlength{\ItemWidth}{2ex}
  \setlength{\ItemHeight}{5ex}
  \setlength{\ItemSpacing}{1ex}
  \coordinate (partial res anchor) at ($(postprocessing label.center -| subbatch label.east) + (2.5ex, \ItemHeight * 2.5 - 2.3ex - 0.5ex)$);
  \coordinate (partial res anchor) at ($(partial res anchor) + (0, 0.1ex)$);
  \foreach \iRow in {0, ..., 2} {
   \foreach \i in {3, ..., 0} {
    \node[answer={($(partial res anchor) + (\ItemSpacing * \i, 0) + (0, -\ItemHeight * \iRow) + (0.4ex * \iRow, 1.15ex * \iRow)$)}]
     (partial res \iRow-\i) at (0, 0) {A};
   }
   \node[bbox={(partial res \iRow-0) (partial res \iRow-3)}] (partial res \iRow) {};
  }
  \node[bbox={(partial res 0) (partial res 2)}] (partial res) {};
  \node[anchor=north] (partial res label) at ($(partial res.south) + (0, -0.5ex)$) {Partial aggregation results};

  \node[procedure, anchor=west] (evaluating label) at ($(routing label.center -| partial res.east) + (7.5ex, 0)$) {
   \hspace*{0.2ex}
   \shortstack[l]{
    \bfseries Evaluation
   }
  };
  \node[proc num, anchor=north] (step 2) at (evaluating label.north west) {2};

  \coordinate (evaluating hot input) at ($(evaluating label.south west) + (7.5ex, -0.3ex)$);
  \coordinate (hot anchor) at ($(evaluating hot input) + (0, -5ex)$);
  \coordinate (hot north) at ($(hot anchor) + (0, 0.2ex)$);
  \coordinate (evaluating cold input) at ($(evaluating hot input) + (2.2ex, 0)$);
  \coordinate (cold bend) at ($(evaluating cold input) + (0, -2.7ex)$);
  \coordinate (cold anchor) at ($(cold bend) + (4.5ex, -2.7ex)$);
  \coordinate (cold north) at ($(cold anchor) + (1.8ex, 2.6ex)$);
  \coordinate (index mid x) at ($(cold north)!.5!(hot north)$);
  \filldraw (hot north) -- ++(-3ex, -5ex)
   coordinate (hot west) -- ++(3ex, 0)
   coordinate (hot south) -- ++(3ex, 0) -- cycle;
  \filldraw[black!60] (cold north) -- ++(4ex, -7ex)
   coordinate (cold east) -- ++(-4ex, 0)
   coordinate (cold south) -- ++(-4ex, 0) -- cycle;
  \node[anchor=north, font=\footnotesize] (hot label) at ($(hot south) + (0, 0.5ex)$) {Hot};
  \node[anchor=north, font=\footnotesize] (cold label) at ($(cold south) + (0, 0.5ex)$) {Cold};
  \node[bbox={(hot north) (hot west) (hot label) (cold north) (cold east) (cold label)}] (index) {};
  \node[inner sep=0.5ex, anchor=north] (index label) at (index mid x |- index.south) {\shortstack{
   Index structures\\
   (\BpTrees)
  }};

  \coordinate (res anchor) at ($(batch anchor |- postprocessing label.center) + (0, \ItemHeight / 2 - 0.5ex)$);
  \coordinate (res anchor) at ($(res anchor) + (0, 0.1ex)$);
  \foreach \i in {8, ..., 0} {
   \node[answer={($(res anchor) + (\ItemSpacing * \i, 0)$)}] (res \i) at (0, 0) {A};
  }
  \node[bbox={(res 0) (res 8)}] (res) {};
  \node[anchor=north] at ($(res.south) + (0, -0.5ex)$) {Results for the batch};

  \begin{scope}[on behind layer]
   \coordinate (routing input) at ($(routing label.west) + (-0.9ex, 0)$);
   \draw[data flow={blue!65}, rounded corners=2.5ex] ($(routing table.north west) + (3.2ex, 0.3ex)$) |- (routing input);
   \draw[data flow={blue!65}] ($(batch.east) + (0.3ex, 0)$) -- (routing input);

   \coordinate (routing output) at ($(routing label.east) + (0.3ex, 0)$);
   \coordinate (evaluating input) at ($(evaluating label.west) + (-0.9ex, 0)$);
   \draw[data flow={blue!75!red!50}, rounded corners=2.5ex] (routing output) -| ($(raq ids.north) + (0, 0.8ex)$);
   \coordinate (subbatch branch) at ($(subbatch 0.west) + (-1.8ex, 0)$);
   \draw[data flow={blue!75!red!50}] (routing output) -- (evaluating input);

   \coordinate (evaluating output) at ($(evaluating label.south west) + (1.5ex, -0.3ex)$);
   \coordinate (postprocessing input) at ($(postprocessing label.east) + (0.3ex, -0.5ex)$);
   \draw[data flow={blue!25!red!50}, rounded corners=2.5ex] ($(raq ids label.south) + (0, 0.7ex)$) |- (postprocessing input);
   \coordinate (partial res branch) at ($(partial res 2.west) + (-2.9ex, 0)$);
   \draw[data flow={blue!25!red!50}] (evaluating output) |- (partial res 0.center) -| (partial res branch) -- (postprocessing input);

   \draw[data flow={blue!75!red!50}] (hot anchor) -- (evaluating hot input);
   \draw[data flow={blue!75!red!50}, rounded corners=0.7ex] (cold anchor) -- (cold bend) -- (evaluating cold input);

   \draw[data flow={red!65}] ($(postprocessing label.west) + (-0.5ex, -0.5ex)$) -- ($(res.east) + (0.5ex, 0)$);
  \end{scope}

  \coordinate (cpu contents south west) at ($(routing table.west |- partial res label.base) + (-0.5ex, +1.4ex)$);
  \node[fit={(cpu contents south west) (routing label) (postprocessing label) (step 3) (routing table) (raq ids label)},
        inner ysep=1ex, inner xsep=0.7ex] (cpu contents) {};
  \begin{scope}[on background layer]
   \fill[green!25] (cpu contents.north west) rectangle (cpu contents.south east);
   \draw[green!70!black, rounded corners=1.2ex, line width=0.8ex] (cpu contents.north west) rectangle (cpu contents.south east);
   \node[inner sep=0.7ex, font=\bfseries, anchor=south east] (cpu label) at ($(cpu contents.south east) + (0.4ex, -0.3ex)$) {CPU};
   \fill[green!70!black] (cpu label.north east)
    [rounded corners=1.2ex] -- (cpu label.south east)
    [sharp corners] -- (cpu label.south west)
    |- cycle;
   \node[text=white, font=\bfseries] at (cpu label) {CPU};
  \end{scope}

  \coordinate (dpu contents west) at ($(step 2.west) + (-0.4ex, 0)$);
  \coordinate (dpu contents east) at ($(index.east) + (1.8ex, 0)$);
  \node[fit={(dpu contents west) (dpu contents east) (index) (index label) (evaluating label) (step 2)},
        inner ysep=1ex, inner xsep=1.1ex] (dpu contents) {};
  \begin{scope}[on background layer]
   \foreach \iDPU in {2, ..., 0} {
    \coordinate (dpu north west) at ($(dpu contents.north west) + (-1.3ex * \iDPU, -1.5ex * \iDPU)$);
    \coordinate (dpu south east) at ($(dpu contents.south east |- partial res label.north) + (-1ex * \iDPU, -1ex * \iDPU + 2.4ex)$);
    \draw[white, rounded corners=1.2ex, line width=1.2ex] (dpu north west) rectangle (dpu south east);
    \fill[green!25] (dpu north west) rectangle (dpu south east);
    \draw[green!70!black, rounded corners=1.2ex, line width=0.8ex] (dpu north west) rectangle (dpu south east);
    \node[inner sep=0.7ex, font=\bfseries, anchor=south east] (dpu label) at ($(dpu south east) + (0.4ex, -0.3ex)$) {DPU};
    \fill[green!70!black] (dpu label.north east)
     [rounded corners=1.2ex] -- (dpu label.south east)
     [sharp corners] -- (dpu label.south west)
     |- cycle;
    \node[text=white, font=\bfseries] at (dpu label) {DPU};

    \ifthenelse{\NOT \iDPU = 0}{
     \draw[data flow={blue!75!red!50}] (routing output) -| (subbatch branch |- subbatch \iDPU.center) -| (evaluating input);
     \ifthenelse{\iDPU = 2}{
      \draw[data flow={blue!25!red!50}] (evaluating output) |- (partial res branch) -- (postprocessing input);
     }{
      \draw[data flow={blue!25!red!50}] (evaluating output) |- (partial res \iDPU.center) -| (partial res branch) -- (postprocessing input);
     }
    }{}
   }
  \end{scope}
 \end{tikzpicture}
 \caption{Workflow of \BpForest processing an RAQ batch.}
 \label{fig:RAQ-batch-flow}
\end{figure*}
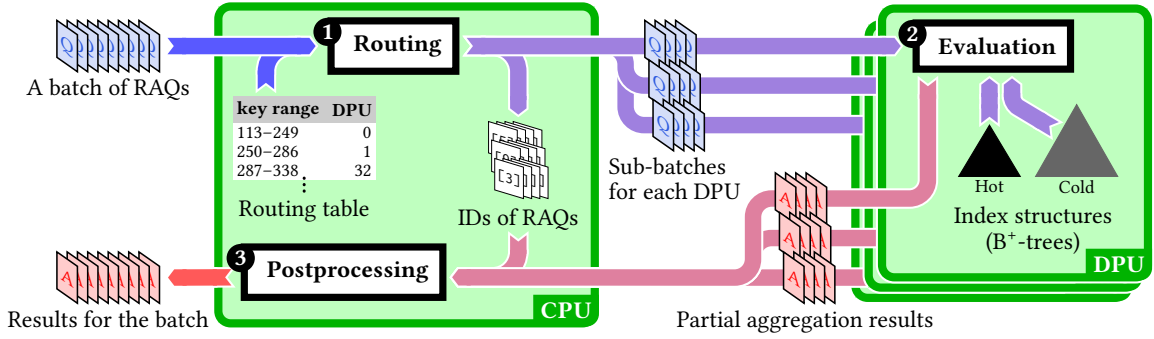

\subsection{Query Processing}
\BpForest takes a batch of RAQs represented as an array of RAQs, and
yields an array of the RAQ results arranged in the same order as the
batch.  \Cref{fig:RAQ-batch-flow} illustrates the flow of processing a
batch of RAQs, which consists three steps: 1) Routing, 2) Evaluation,
and 3) Postprocessing.  In the following, we describe each step one by
one.

\paragraph{Routing}
Given a batch of RAQs, the host CPU creates a sub-batch for each DPU
according to the routing table.  Each sub-batch contains the RAQs, of
which the ranges intersect with the key ranges of the DPU.  For
postprocessing, it also records the IDs (or indexes) of the RAQs in each
sub-batch, in order, in the corresponding side table.  After that, the
sub-batches are sent to the DPUs.

 The The CPU initially divides a given batch
evenly among all the threads.  Each thread has its local arrays for
sub-batches and records the RAQs into the appropriate local arrays
according to the routing table.  After that, the local arrays of all
threads are gathered, producing sub-batches.

\paragraph{Evaluation}
Given a sub-batch of RAQs, the DPU calculates the partial result of each
RAQ through aggregation over both \BpTrees.  The partial results are
stored in a resultant array in order.

This step is simply parallelized by evenly dividing the sub-batch among
all threads on the DPU.  No synchronization is required.

\paragraph{Postprocessing}
The CPU collects the partial results from the DPUs, aggregates them, and
stores the results into an output array in the order of the input RAQs.
To store each result in order in constant time, the side table of RAQ
IDs for each sub-batch is used.

 To pursue synchronization-free aggregation,
we assign the partial results of an RAQ to the thread that has routed
it; specifically, we scatter all the partial results to thread-local
arrays, as the inverse of gathering in routing.  After that, all threads
run independently.

\section{Evaluation}
\label{sec:eval}

We experimentally demonstrate the following points:

\begin{itemize}
 \item Query density-driven partitioning outperforms density-unaware
       partitioning in spatiotemporal load balancing for skewed queries
       (\cref{sec:eval:coc}).
 \item Query density-driven partitioning is adaptable to various memory
       limits and DPU numbers (\cref{sec:eval:sens}).
 \item The impacts of the downsides of our approach for workload
       sensitivity (\cref{sec:eval:wrk_sens}).
 \item \BpForest has performance comparable to the state-of-the-art
       skew-resistant index \PimTree~\cite{kang22:PIM-tree} in point-get
       queries, yet handles RAQs efficiently
       (\cref{sec:eval:vs_PIM-tree}).
\end{itemize}

\subsection{Experimental Setting}
\label{sec:eval:setting}

The workload used in our experiments, unless otherwise noted, is as
follows.  The dataset consisted of 500 million key-value pairs and was
distributed to DPUs (i.e., $D = \SI{500}{M}$).  The keys and values were
 integers synthesized to be uniformly distributed within the 64-bit space.  The upper
limit on the number of key-value pairs assigned to each DPU was $1.1
\times \frac{D}{P}$.  This is achieved by setting $\alpha = 10$ in query
density-driven partitioning.  Queries on these data were given in
batches of 1 million RAQs (i.e., $Q = \SI{1}{M}$).  As an instance of
RAQs, we used range-count-if queries, each of which consisted of a key
range $r_q$ and a value $v_q$, and counted the number of key-value pairs
$(k,v)$ such that $k \in r_q$ and $v = v_q$.  The lower ends of $r_q$
followed the Zipf-composite
distribution~\cite{Gilad20:EvenDB,zhong21:REMIX,chen22:flex_addr_space}:
the upper 14 bits followed a Zipf distribution, and the lower 50 bits
followed a uniform distribution.  We examined variously skewed queries
by varying the Zipf skewness parameter\footnote{The default setting of
the Zipf skewness in YCSB is 0.99; the skewness is known to
fall typically into the range [0.2,1.2]~\cite{yang16:zipf_evidence}.}
between 0.6 and 1.2; larger
values indicate more severe skewness.  The upper ends of $r_q$ were
designed to encompass approximately 100 key-value pairs each.  We
prepared 30 batches of queries for each experimental parameter setting,
used the first 10 batches for warming up, and used the remaining 20
batches for measurement.  The error bars in the subsequent figures
represent the $1\sigma$ intervals calculated from the results of the 20
evaluation batches.  The first warm-up batch was also used as the
reference workload for partitioning.

We define \emph{imbalance factors} as indicators of spatial and temporal
load balance, which we use in the following subsections.  The imbalance
factor of data, indicating spatial load imbalance, is calculated by
dividing the maximum number of key-value pairs stored in each DPU by the
arithmetic mean.  The imbalance factor of queries, indicating temporal
load imbalance, is calculated by dividing the maximum number of queries
processed by each DPU by the arithmetic mean.  By definition, imbalance
factors are greater than or equal to 1; values closer to 1 are better.

The hardware environment was a server purchased from
UPMEM\textregistered{}, equipped with a pair of Intel\textregistered{}
Xeon\textregistered{} Silver 4216 (16 cores, 32 threads, \SI{2.1}{GHz},
\SI{22}{MB} cache), \SI{256}{GB} memory on the CPU side, and the
PIM-enabled memory that consisted of 40 ranks of UPMEM\textregistered{}
v1B, where each rank had up to 64 DPUs (see \cref{sec:upmem} for the
details).  Excluding inactive ones, a total of 2,538 DPUs were available
for use.  Unless otherwise noted, we used 16 ranks, i.e., 1,012 DPUs ($P
= 1012$) and ran 16 threads on each DPU.

We prepared two competitors: a \BpTree variant with query
density-unaware partitioning and \PimTree.  As a representative of query
density-unaware partitioning, we used the partitioning that minimizes the
imbalance factor of queries under a spatial constraint such that the
imbalance factor of data did not exceed the given upper limit of 1.1 and
the number of  partitions was the same as that of DPUs.
 This is also
known as a solution to the chains-on-chains partitioning
problem~\cite{pinar04:optimal_1d_part}.
The \PimTree implementation was based on the publicly available
version~\cite{kang22:PIM-tree} by the authors, with minimal
modifications\footnote{They have no significant effect on performance as
far as we know.} to use a specified batch size.  It should be noted
that, to ensure a fair comparison, the experiments in
\cref{sec:eval:vs_PIM-tree} used the parameters well-suited to
\PimTree: 1M was the default batch size of \PimTree; 11 threads on each
DPU, instead of 16, were used because that was the maximum feasible number for \PimTree; and all the 2,538
DPUs in 40 ranks, instead of 1,012, were used because \PimTree could only
use all the DPUs in the machine.  More importantly, \BpForest
and \PimTree are not directly comparable for range queries; \BpForest
supports RAQ, whereas \PimTree supports range-scan (or simply scan)
queries with query fusion.  Only for get queries are they comparable.
To exclude the warm-up for \PimTree,
we calculated the time for 20 batches by subtracting the time for 10
batches from the time for 30 batches.

\subsection{Query Density-Driven vs. Density-Unaware}
\label{sec:eval:coc}

\begin{figure}[tb]
 \centering
 \includegraphics{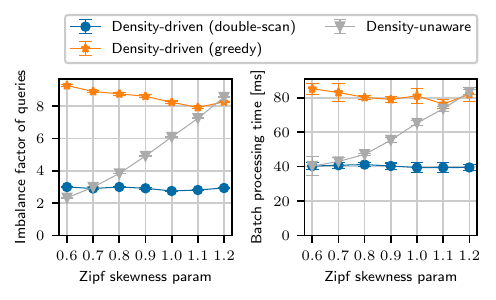}

 \begin{minipage}{0.565\linewidth}
  \subcaption{Imbalance factor of queries.}
  \label{fig:vs_coc:queries}
 \end{minipage}
 \begin{minipage}{0.408\linewidth}
  \subcaption{Elapsed time.}
  \label{fig:vs_coc:time}
 \end{minipage}

 \caption{The imbalance factor of queries and the processing time of RAQ
 batches for query density-driven partitioning (with two hot range
 selection algorithms) and query density-unaware partitioning.}
 \label{fig:vs_coc}
\end{figure}

We compare the spatiotemporal load balance resulting from query
density-driven partitioning with that from query density-unaware
partitioning under various levels of skewness.
\Cref{fig:vs_coc:queries} illustrates the relationship between query
skewness and the imbalance factor of queries.  With query
density-unaware partitioning, the imbalance factor increased as the
queries became more skewed.  Meanwhile, query density-driven
partitioning with double-scan achieved a nearly constant imbalance
factor.  These results indicate that query density-unaware partitioning
lost skew resistance under memory constraints, whereas query
density-driven partitioning retained it.

It has also been shown that query density-driven partitioning with
double-scan resulted in a smaller imbalance factor than query
density-unaware partitioning when the Zipf skewness
was around 1.0.  This result justifies the second scan to produce a
sufficient number of hot partitions and reduce idle DPUs.
 At a Zipf skewness of 0.6, query density-unaware partitioning
achieved a better load balance than query density-driven partitioning.
This is attributed to the design of our method, which is intended to
maintain load balance across a range of query skews
rather than to optimize the balance further under sufficiently mild skew.

Such a level of spatiotemporal load balancing cannot be yielded with
simple one-dimensional partitioning shown in
\cref{fig:example-partitions:equal-data,fig:example-partitions:equal-query}.
For example, at the Zipf skewness of 1.0, data-balanced partitioning
results in an imbalance factor of queries of 334, and query-balanced
partitioning results in an imbalance factor of data of 10.2.

\Cref{fig:vs_coc:time} illustrates the relationship between query
skewness and RAQ-batch processing time.  In most cases, the processing
time had the same trend as the imbalance factor of queries shown in
\cref{fig:vs_coc:queries}.  It indicates that, under realistic query
skewness, query density-driven partitioning improved query load balance,
thereby improving the performance of RAQ processing.  We conducted a
similar experiment with range-max queries and observed the same trend.

\subsection{Parameter Sensitivity}
\label{sec:eval:sens}

\begin{figure}[tb]
 \centering
 \includegraphics{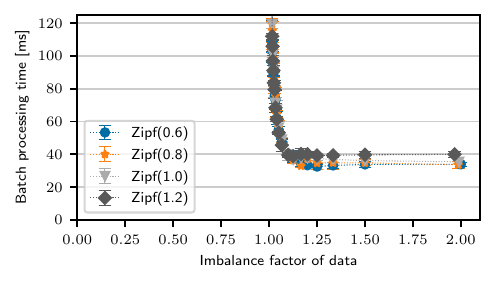}
 \caption{Trade-off between the imbalance factor of data and RAQ-batch
 processing time: plotting by varying the partitioning parameter $\alpha
 = 1, 2, 3, 4, 5, 6, 8, 10, 15, 20, \dots, 65$.}
 \label{fig:alpha}
\end{figure}

We evaluate the impact of changes in memory usage and the numbers of DPUs on
query density-driven partitioning.

\Cref{fig:alpha} illustrates the relationship between the imbalance
factor of data and RAQ-batch processing time.  We adjust the imbalance
factor of data by changing the partitioning parameter $\alpha$.
Regardless of the Zipf skewness, we have confirmed that the processing
time became small (toward the bottom) by relaxing the imbalance factor
(toward the right).  Therefore, tuning $\alpha$ as per the
available memory capacity leads us to the sweet spot, regardless of the
workload.

As the imbalance factor of data increased, the RAQ-batch processing time
converged to different values depending on query skewness.  We attribute
it to data chunking (\cref{sec:load_balance:overview}).  Since load
balancing is based on distributing data chunks, we cannot help but have
load imbalance in the case where one chunk receives a large number of
queries, which becomes more significant in the case where chunk
granularity becomes coarser or query skewness becomes more severe.  Note
that the chunk granularity in our experiments was derived from the
\BpForest implementation, not from the query density-driven partitioning
scheme.  The negative effect of chunking may therefore be mitigated,
depending on the data structures and chunking methods used.

\begin{figure}[tb]
 \centering
 \includegraphics{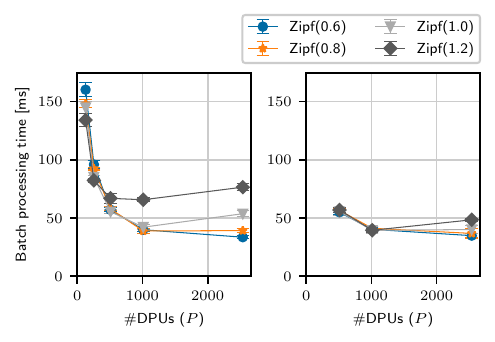}

 \hspace*{0.111\linewidth}
 \begin{minipage}{0.4\linewidth}
  \subcaption{$D = \SI{200}{M}$.}
  \label{fig:nr_dpus:200M}
 \end{minipage}
 \hspace*{0.015\linewidth}
 \begin{minipage}{0.44\linewidth}
  \subcaption{$D = \SI{500}{M}$.}
  \label{fig:nr_dpus:500M}
 \end{minipage}

 \caption{Relationship between the number of DPUs ($P$) and RAQ-batch processing time.}
 \label{fig:nr_dpus}
\end{figure}

\Cref{fig:nr_dpus} illustrates the relationship between the number of
DPUs $P$ and batch processing time.  We here also conducted experiments
with $D = \SI{200}{M}$ because the key-value pairs did not fit into
memory with a small $P$ and a data size of $D = \SI{500}{M}$.  Except
for the cases of Zipf(1.2) with $D = \SI{200}{M}$, we observed that the
querying performance scaled up to $P = 1,012$, and had no gain or even
deteriorated in $P = 2,538$.  That was due to the negative effect of
data chunking on load balancing, in addition to the augmented cost of
query routing and postprocessing on the CPU, which has been corroborated
by the fact that performance degraded more with smaller data sizes and
more severe query skewness.  It is evident that the cases of Zipf(1.2)
with $D = \SI{200}{M}$ have suffered severely from this phenomenon.

\subsection{Workload Sensitivity}
\label{sec:eval:wrk_sens}

We investigated the impact of deviation from the reference workload.  We
set the reference workload as a Zipf-composite with a Zipf skewness
parameter of 1.0.  We selected a base partition and amplified the relative probability of queries
reaching that base partition by a factor of 10, thereby deviating
from the reference workload.  We refer to the resulting
distribution as ZcAmp($b_i$, 10), where $b_i$ is the base partition
subject to amplification.

\begin{figure}[tb]
 \centering
 \includegraphics{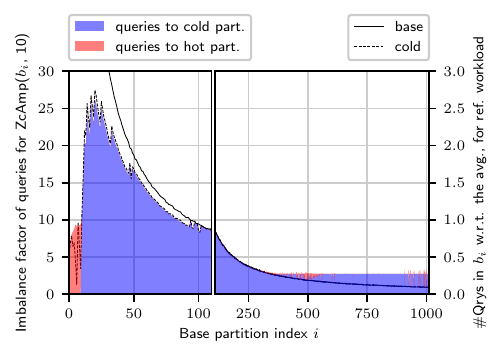}

 \caption{Area plot of the imbalance factor of queries for ZcAmp($b_i$, 10) versus the base partition index $i$ (blue/red indicating queries to cold/hot partitions; left vertical axis), overlaid with a line plot of the number of queries to $b_i$ and its cold partitions prior to deviation, normalized by the average number across DPUs (right vertical axis).}
 \label{fig:wrk_sens}
\end{figure}

The area plot in \cref{fig:wrk_sens} shows the imbalance factor of
queries for the ZcAmp($b_i$, 10) workload for each base partition $b_i$,
where the partitions were computed based on the reference workload before amplification.
Colors indicate
the breakdown of query destinations (hot or cold partitions) in the most
loaded DPU.  The line plot shows the number of queries to each base
partition before the deviation, relative to the average number of
queries among DPUs.  The dotted line indicates the portion of queries
routed to cold partitions.

These results show that, when the workload deviates from the reference workload, cold partitions that attract increased queries are more likely to become performance bottlenecks than hot partitions.
More precisely, when the query density increases in cold partitions that originally had relatively high density, load imbalance arises.
In fact, for $10 \le i \le 331$, the DPU that hosted the cold partitions of the base partition $b_i$ processed the largest number of queries for ZcAmp($b_i$, 10).
This is indicated in \cref{fig:wrk_sens}, where the shape of the area chart in blue matches the dotted line for such $i$.
In particular, the DPU that hosted the cold partitions of $b_{20}$ processed the largest number of queries before the workload deviation, and the imbalance factor of queries increased from 2.8 to 26 when we amplify the query density in this base partition $b_{20}$ (i.e., for ZcAmp($b_{20}$, 10)).
These results are consistent with our analysis in \cref{sec:load_balance:double}, which shows that the total query load on cold partitions of a single base partition can exceed $\frac{Q}{P}$  (and bounded by $\frac{\alpha + 1}{3}\frac{Q}{P}$) whereas the query load on a hot partition is at most $\frac{Q}{P} + M_q$.

\subsection{\BpForest vs. \PimTree}
\label{sec:eval:vs_PIM-tree}

\begin{figure}[tb]
 \centering
 \includegraphics{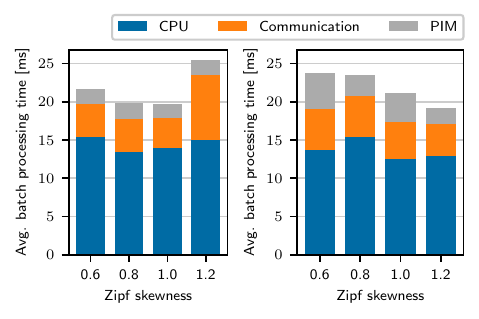}
 \vspace{-\bigskipamount}
 \\
 \begin{minipage}[t]{0.5\linewidth}
  \centering
  \subcaption{\BpForest.}
  \label{fig:vs-pimtree:bpforest}
 \end{minipage}
 \begin{minipage}[t]{0.45\linewidth}
  \centering
  \subcaption{\PimTree.}
  \label{fig:vs-pimtree:pimtree}
 \end{minipage}
 \caption{Get-batch processing time: \BpForest vs. \PimTree.}
 \label{fig:vs-pimtree}
\end{figure}

\begin{figure}[tb]
 \centering
 \includegraphics{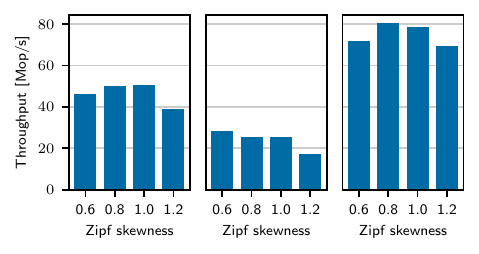}

  \hspace{2.5em}
  \begin{minipage}{0.27\linewidth}
  \centering
  \subcaption{1M Gets.}
  \label{fig:bpforest-throughput:1M_gets}
  \end{minipage}
  \begin{minipage}{0.27\linewidth}
  \centering
  \subcaption{1M RAQs.}
  \label{fig:bpforest-throughput:1M_raqs}
  \end{minipage}
  \begin{minipage}{0.27\linewidth}
  \centering
  \subcaption{20M Gets.}
  \label{fig:bpforest-throughput:20M_gets}
  \end{minipage}

 \caption{Throughput of batched queries on \BpForest.}
 \label{fig:bpforest-throughput}
\end{figure}

We compare the performance of \BpForest with \PimTree for get queries.
\Cref{fig:vs-pimtree} illustrates the arithmetic means of get-batch
processing time for \BpForest and \PimTree, and their breakdowns: CPU,
communication, and PIM, where the PIM time means the duration of the DPU
kernel running.  We measured the PIM and communication times, and
calculated the CPU time by subtracting those values from the total.
\BpForest outperformed \PimTree up to Zipf(1.0), and the
tables were turned at Zipf(1.2).  That was due to
query fusion implemented in \PimTree; more skewed workloads lead to
greater degree of fusion, thereby decreasing the querying cost.  In contrast, as
seen in \cref{fig:nr_dpus:200M}, \BpForest simply suffers a larger
performance penalty from more skewed workloads.  Meanwhile, \BpForest
has benefited from its simplicity and spatial locality.  The
PIM time of \BpForest was generally less than that of \PimTree.
For the CPU time, \PimTree benefited from pipelined querying with two
CPU threads, whereas \BpForest omitted it for implementation simplicity.
Both advantages offset each other, thereby resulting in
comparable performance around Zipf(1.0).

Here, we emphasize the virtues of \BpForest for two major points.
One is that \BpForest supports efficient RAQs.  As illustrated in
\cref{fig:bpforest-throughput:1M_gets,fig:bpforest-throughput:1M_raqs},
batched RAQs had no less than 50\% of throughput as get-batch up to
Zipf(1.0); even at Zipf(1.2), it was 44.2\%.
The ratio of throughput shrinkage is a very small constant
compared to the number of hit items, which is about 100.  The other is
that \BpForest is capable of handling significantly larger batches of
queries than \PimTree; while \PimTree could handle up to 1M without
modification and up to 5M with minor modifications, \BpForest could
handle 1.1G. Consequently, the throughput was enhanced, and the
performance deterioration at Zipf(1.2) was suppressed, as illustrated in
\cref{fig:bpforest-throughput:20M_gets}.  This is a blessing of its simplicity
and spatial locality.

\begin{figure}[tb]
 \centering
 \includegraphics{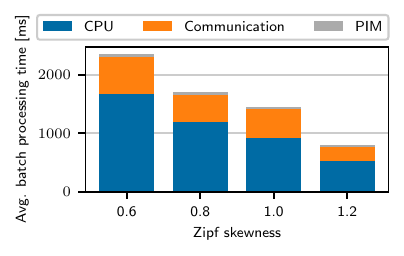}
 \caption{Scan-batch processing time on \PimTree.}
 \label{fig:vs-pimtree-scan}
\end{figure}

Lastly, saying again, a direct comparison of the results of \BpForest
and \PimTree (\cref{fig:vs-pimtree-scan}) is not intended because
they had different design choices and handled different queries.
However, if we would like to implement RAQs for databases backed by
\PimTree, it is natural to issue scan queries and aggregate the results
on the CPU.  In that scenario, the scan-batch processing time shown in
\cref{fig:vs-pimtree-scan} is inevitable.  Therefore, the
results have demonstrated that \BpForest works well for RAQs.

\section{Related Work}
\label{sec:related}
In ordered-key indexing, key-range partitioning is a natural
choice~\cite{choe19:cds_ndp,choe22:HybriDS} amenable to the hardware of
UPMEM PIM, while it suffers considerably from skewed queries. In
contrast, hash-based random distribution of elements over DPUs is
effective in load balancing against skewed queries, which is the heart
of PIM-balanced skip lists~\cite{kang21:PIM_model}.  In a PIM-balanced
skip list, the lower part is randomly distributed over DPUs, and the
upper part is replicated among DPUs.  In querying, all DPUs
cooperatively perform pointer chasing by passing the remaining traversal
queries to other DPUs that own the referents.
\PimTree~\cite{kang22:PIM-tree}, which has been compared with \BpForest in
\cref{sec:eval}, is an improved version of PIM-balanced skip lists in
terms of two major points.  One is the push-pull search strategy: it
combines the pull strategy on the CPU, which efficiently handles a small
set of nodes fitting in the cache, and the push one on the DPUs, which
exploits the massive parallelism of DPUs.  It adopts the pull strategy
to handle high-contention (i.e., hot) internal nodes in search, where
hot nodes are temporarily copied to the CPU side, more specifically,
into the cache.  The
other is chunking nodes of skip lists to take advantage of locality in
search, as in B-trees.  The chunk size of \PimTree is 16 by default; it
is reasonably small to suppress load imbalance against skewed point
queries, whereas it increases the probability that the ranges of range
queries lie over multiple DPUs, incurring the amplification of range
queries.

The state-of-the-art PIM-oriented skip list RADAR~\cite{hua24:RADAR}
alleviates a locality sacrifice in the random distribution of fine-grained,
even chunks, as in \PimTree, by using the hotness information. It
initially assigns a large key range to each DPU to enhance spatial
locality.  Then, it identifies busy DPUs at runtime, adaptively splits
the key ranges on them into finer-grained subranges, and redistributes
them and their key-value pairs to all DPUs.  Finest-grained subranges
are migrated to the CPU side to exploit the cache locality. A key
improvement over \PimTree is that RADAR biases the random distribution
of key subranges to hot ranges; hotter ranges are split finer so that
every DPU evaluates queries over hot ranges in their DRAM. Some basic ideas of
enhancing locality while taming
skewness, such as adaptive partitioning based on hotness, are shared
between RADAR and \BpForest.  However, RADAR addresses only load
balancing of query evaluation and does not account for the imbalance of
data distribution.  Moreover, it is not designed for general
range-aggregate queries.
Note that RADAR's adaptive online partitioning is monotone and
irreversible, thereby not addressing workload changes.

PIMLex~\cite{cui25:PIMLex} is a state-of-the-art PIM-oriented learned
index designed to run efficiently on UPMEM PIM, where learned
indexes~\cite{sun25:learned_index,liu25:multi-dim_learned_index} are a
kind of workload-driven approach and calculate item positions via
machine learning models.  Although PIMLex also
addresses skewed queries, there are differences in design between PIMLex
and \BpForest.  In PIMLex, every key-value pair resides on the CPU side,
and a learned index containing part of the keys resides on the DPU side;
a single partition of the key space is assigned to each DPU.  PIMLex
thus does not support range-aggregate queries and is not designed to run
key-value pair aggregation on DPUs.  PIMLex addresses skewed queries by
assigning replicas of hot partitions to multiple DPUs and randomly
selecting the DPUs on querying.  The replication of partitions
complicates query processing in the case where key-value pairs reside on
the DPU side, as in \BpForest.

Even though typical workload-driven approaches to indexing, including
our work, do not address dynamic workload changes, there has recently
been a notable exception: HeatList~\cite{shen25:HeatList}, which is a
range index based on skip lists and is aware of hotspots (i.e., hot
partitions in this paper).  It is equipped with a fast reaction
mechanism that quickly detects hotspot shifts and moves nodes to the top
layer, which is designed to fit into the CPU cache.  HeatList is
designed for shared memory and does not address load balancing.  The
problem setting differs from our work.

Load-balanced partitioning of one-dimensional arrays, which is called
the chains-on-chains partitioning problem, has been well studied.  There
are efficient exact algorithms~\cite{pinar04:optimal_1d_part} for
optimal partitioning, which are inherently sequential and known to be
difficult to run efficiently in
parallel~\cite{lieber14:scalable_1d_part}.  The variant of
chains-on-chains partitioning used in \cref{sec:eval} runs in serial and
yields optimal solutions in the size-constrained setting.  Our
partitioning scheme is designed for a relaxed setting of the size-constrained
chains-on-chains partitioning problem in which the total number of
partitions is not exactly $P$ but is bounded by $O(P)$.  Our
partitioning algorithms are designed to run in parallel as-is
and are PIM-friendly because hot range selection can run independently
on each DPU.

PIM systems are currently under active development in both academia and
industry.  Another representative of commercial PIM systems, other than
UPMEM PIM, is Samsung HBM-PIM~\cite{lee21:HBM-PIM}.  It adopts a
lockstep computational model and is apt to regular data-parallel
workloads, such as tensor processing in deep neural networks.  Although
both are classified as PIM, UPMEM PIM and HBM-PIM have different,
complementary computing scopes.

\section{Conclusions}
We have presented a query density-driven partitioning scheme of
ordered-key indexes for spatiotemporal load balancing on
resource-constrained PIM systems like UPMEM PIM.  It balances query
density among
DPUs by distributing high query-density ranges over DPUs, while the
remaining part is distributed contiguously and mostly evenly. We have
developed \BpForest, which supports efficient range-aggregate queries on
UPMEM PIM, by applying our partitioning scheme to \BpTrees.  The
experimental results have shown that \BpForest exhibited higher skew
resistance than a \BpTree based on density-unaware partitioning and
performance comparable to \PimTree, the state-of-the-art index without
native support of RAQs, on point-get queries.

There remains much room for further development of our partitioning
scheme.  First and foremost is greater sophistication to address dynamic
workload changes.  In the concurrent
paper~\cite{hideshima26:partial_resharding}, we have already developed
an online partial rebalancing method and obtained preliminary
experimental results showing that it is more efficient than the full
rebalancing method used in this paper under dynamically changing
workloads.  Thus, we can suppose a variety of rebalancing methods based
on our partitioning scheme.  Organizing and coordinating them in a
database system to holistically handle realistic changes in workload is
a significant challenge.

Last but not least, although UPMEM PIM systems were discontinued in
2025, this does not imply that the UPMEM PIM model is useless; we can
identify similar system models, including MPPA architectures.  We
therefore consider that our work on techniques to adapt data-intensive
computing to the UPMEM PIM model is a general, reusable effort to break
the memory wall.  The adaptation to MPPA architectures, such as
Kalray~\cite{dupont15:Kalray} and PEZY~\cite{hatta25:PEZY-SC4s}, holds
considerable promise.

\begin{acks}
 This work was supported by JSPS KAKENHI Grant Number 22K17872, 23K24822, 24KJ0638, 26K02887.
\end{acks}

\balance

\bibliography{reference}

\appendix

\begin{Proof-in-Appendix}
\section{Proof of the Properties of the Greedy Hot Selection Algorithm}
\label{sec:proof:greedy}
We prove the properties satisfied by the result of
\cref{algo:hot_select} below.  Unless otherwise specified, line numbers
are used to refer to the processing and states within
\cref{algo:hot_select}.

\begin{theorem}
 \label{thm:greedy:hot_size}
 The data load on each hot partition is less than
 $\frac{1}{\alpha}\frac{D}{P} + M_d$.
\end{theorem}

\begin{proof}
 Suppose that there is a hot partition $h_? \coloneqq \{c_i\}_{i=l_?}^r$
 such that
 \begin{equation}
  \SizeOf(h_?) \ge \frac{1}{\alpha}\frac{D}{P} + M_d.  \label{eq:greedy:hot_size:supposed}
 \end{equation}
 From the definition of $M_d$,
 \begin{equation}
  M_d \ge \SizeOf\left( \{c_{l_?}\} \right).  \label{eq:greedy:hot_size:sizeof_head}
 \end{equation}
 From \cref{eq:greedy:hot_size:supposed,eq:greedy:hot_size:sizeof_head},
 $l_? < r$, and so we can consider a non-empty sequence of data chunks
 $h_?' \coloneqq \{c_i\}_{i=l_?+1}^r$.  Because it follows from
 \cref{eq:greedy:hot_size:supposed,eq:greedy:hot_size:sizeof_head} that
 $\SizeOf(h_?') = \SizeOf(h_?) - \SizeOf\left( \{c_{l_?}\} \right)
 \ge \frac{1}{\alpha}\frac{D}{P}$, after \cref{algo:hot_select} passes
 \lineref{algo:hot_select:window} with such $r$, $l' \ge l_? + 1 >
 l_?$ and after the next line $l > l_?$.  This contradicts the assumption
 that $h_?$ was selected as a hot range in
 \lineref{algo:hot_select:hot}.
\end{proof}

\begin{theorem}
 \label{thm:greedy:hot_nqry}
 The query load on each hot partition is between $\frac{Q}{P}$ and
 $\frac{Q}{P} + M_q$.
\end{theorem}

\begin{proof}
 Since a range selected as a hot range should pass the condition in
 \lineref{algo:hot_select:thres}, the query load is $\frac{Q}{P}$ or
 higher.  Suppose that there exists a hot partition $h_? \coloneqq
 \{c_i\}_{i=l_?}^{r_?}$ such that
 \begin{equation}
  \NQrys(h_?) \ge \frac{Q}{P} + M_q.  \label{eq:greedy:hot_query:supposed}
 \end{equation}
 From the definition of $M_q$,
 \begin{equation}
  M_q \ge \NQrys\left( \{c_{r_?}\} \right).  \label{eq:greedy:hot_query:nqry_last}
 \end{equation}
 From \cref{eq:greedy:hot_query:supposed,eq:greedy:hot_query:nqry_last},
 $l_? < r_?$, and so we can consider a non-empty sequence of data chunks
 $h_?' \coloneqq \{c_i\}_{i=l_?}^{r_?-1}$.  Here, because $l$ is updated only in
 lines \ref{algo:hot_select:l} and \ref{algo:hot_select:next_scan}, $l$ increases monotonically during the
 algorithm.  The fact that $h_?$ was selected as a hot range in
 \lineref{algo:hot_select:hot} means that when it reached
 \lineref{algo:hot_select:l} with $r = r_?$, the value of $l$ became
 $l_?$.  Therefore, based on monotonicity, when $r = r_? - 1$,
 $l \le l_?$.  At this point,
 \begin{align*}
  \NQrys\left( \{c_i\}_{i=l}^r \right) &= \NQrys\left( \{c_i\}_{i=l}^{r_?-1} \right) \\
  &\ge \NQrys(h_?')  \quad (\because l \le l_?) \\
  &= \NQrys(h_?) - \NQrys(\{c_{r_?}\}) \\
  &\ge \frac{Q}{P},  \quad (\because \text{\cref{eq:greedy:hot_query:supposed,eq:greedy:hot_query:nqry_last}})
 \end{align*}
 so $\{c_i\}_{i=l}^{r_?-1}$ should pass the condition in
 \lineref{algo:hot_select:thres} and be selected as a hot range.
 Then, after $l$ is updated to $r_?$ in
 \lineref{algo:hot_select:next_scan}, the next iteration begins, which
 is the only iteration in \cref{algo:hot_select} where $r$ has the value
 $r_?$.  However, if this is the case, due to the monotonicity of $l$,
 $l \ge r_?$ holds thereafter, so $l = l_?$ cannot be true, and $h_? =
 \{c_i\}_{i=l_?}^{r_?}$ cannot be selected as a hot range in
 \lineref{algo:hot_select:hot}.  This contradicts the assumption, so the
 query load of each hot partition is less than $\frac{Q}{P} + M_q$.
\end{proof}

\begin{definition}
 A sequence of hot ranges $\left\{ \{c_i\}_{i=l_j}^{r_j} \right\}_{j=1}^{n}$ is \emph{consecutive} if for all $j \in [1, n)$, $r_j + 1 = l_j$.
\end{definition}
\begin{definition}
 A consecutive sequence of hot ranges\break $S \coloneqq \left\{ \{c_i\}_{i=l_j}^{r_j} \right\}_{j=1}^{n}$ is \emph{maximal} if there exists no consecutive sequence $T \ne S$ such that $S$ is a subseqeunce of $T$.
\end{definition}

\begin{lemma}
 \label{thm:greedy:hot_lump}
 A maximal sequence of hot ranges in the key space starts from
 the starting point of a base partition or has a total data size of
 $\frac{1}{\alpha}\frac{D}{P}$ or more.
\end{lemma}

\begin{proof}
 When $l$ is updated to $l'$ in \lineref{algo:hot_select:l},
 $\SizeOf\left(\{c_i\}_{i=l}^r\right) \ge \frac{1}{\alpha}\frac{D}{P}$
 holds due to \lineref{algo:hot_select:window}.  Otherwise,
 $\SizeOf\left(\{c_i\}_{i=l}^r\right)$ is larger than in the previous
 iteration because $r$ has incremented.  Therefore, once $l$ is updated in \lineref{algo:hot_select:l},
 $\SizeOf\left(\{c_i\}_{i=l}^r\right) \ge \frac{1}{\alpha}\frac{D}{P}$
 holds until the algorithm finds a hot range.

 Here, considering a maximal sequence of hot ranges that does not start from
 the starting point of the base partition, the value of $l$ must be
 updated in \lineref{algo:hot_select:l} while finding the first hot range $h_0$ of the sequence.
 This is because otherwise, the beginning of $h_0$ would be that of
 the base partition or immediately after the previous hot range.
 Therefore, $h_0$ satisfies $\SizeOf(h_0) \ge
 \frac{1}{\alpha}\frac{D}{P}$, and the data size of the entire sequence
 containing it is also greater than or equal to this value.
\end{proof}

\begin{lemma}
 \label{thm:greedy:window_nqry}
 Ignoring queries to the hot ranges, placing the window of
 \FindHot anywhere in a base partition results in the number of queries
 to the data chunks that overlap with the window being less than $\frac{Q}{P}$.
\end{lemma}

\begin{proof}
 Let $b \coloneqq \{c_i\}_{i=s}^e$ be a base partition, $h$ be
 $\FindHot(b)$, $h^*$ be $\bigcup h$, and
 \begin{align}
  L(z) &\coloneqq \left\{l \in [s, z] \relmiddle| \SizeOf\left( \{c_i\}_{i=l}^z \right) \ge \frac{1}{\alpha}\frac{D}{P} \right\},  \label{eq:greedy:window_nqry:L} \\
  a(z) &\coloneqq \begin{cases} s & (L(z) = \emptyset) \\ \max L(z) & (\text{otherwise}) \end{cases}.  \label{eq:greedy:window_nqry:a}
 \end{align}
 Then, this lemma says that for all
 $z$ such that $s \le z \le e$,
 \[
  \NQrys \left( \{c_i\}_{i=a(z)}^z \setminus h^* \right) < \frac{Q}{P}
 \] holds.

 First, we show that $\{c_i\}_{i=a(z)}^z \setminus h^*$ is $\emptyset$ or a
 sequence of consecutive data chunks.
 Note that $a(z) \le z$ by definition, so $\{c_i\}_{i=a(z)}^z$ is well-defined.
 Suppose that it consists of two or more
 sequences.  Then, $a(z) + 2 \le z$, and a maximal sequence of hot ranges is
 included in $\{c_i\}_{i=a(z)+1}^{z-1}$.  Because $a(z) \ge s$, we have $a(z)+1 >
 s$, so by \cref{thm:greedy:hot_lump},
 \begin{equation}
  \SizeOf\left( \{c_i\}_{i=a(z)+1}^{z-1} \right) \ge \frac{1}{\alpha}\frac{D}{P}.  \label{eq:greedy:window_nqry:inside}
 \end{equation}
 When $L(z) = \emptyset$, for all $l \in [s, z]$, $\SizeOf\left(
 \{c_i\}_{i=l}^z \right) < \frac{1}{\alpha} \frac{D}{P}$ by \cref{eq:greedy:window_nqry:L}.  If so,
 \[
  \SizeOf\left( \{c_i\}_{i=a(z)+1}^{z-1} \right) \le \SizeOf\left( \{c_i\}_{i=a(z)}^z \right) < \frac{1}{\alpha} \frac{D}{P},
 \]
 which contradicts
 \cref{eq:greedy:window_nqry:inside}.
 When $L(z) \ne \emptyset$, from
 \cref{eq:greedy:window_nqry:inside},
 \[
  \SizeOf\left( \{c_i\}_{i=a(z)+1}^z \right) \ge \SizeOf\left(\{c_i\}_{i=a(z)+1}^{z-1}\right) \ge \frac{1}{\alpha}\frac{D}{P}
 \]
 holds and implies that $a(z)+1 \in L(z)$,
 $\max L(z) \ge a(z)+1$, but this contradicts \cref{eq:greedy:window_nqry:a}.

 Therefore, either $\{c_i\}_{i=a(z)}^z \setminus h^* = \emptyset$ or
 \begin{equation}
  {}^\exists a', {}^\exists z' \text{ s.t. } \{c_i\}_{i=a(z)}^z \setminus h^* = \{c_i\}_{i=a'}^{z'}  \label{eq:greedy:window_nqry:range}
 \end{equation}
 holds.  In the former case, $\NQrys(\{c_i\}_{i=a(z)}^z \setminus h^*) = 0
 < \frac{Q}{P}$.  Below, we consider the latter case.
 From \cref{eq:greedy:window_nqry:range},
 \begin{gather}
  {}^\forall c \in \{c_i\}_{i=a'}^{z'}, c \notin h^*,  \label{eq:greedy:window_nqry:no_hot} \\
  a(z) \le a' \le z' \le z.  \label{eq:greedy:window_nqry:domain}
 \end{gather}
 In applying \cref{algo:hot_select} to $b$, we should have reached
 \lineref{algo:hot_select:window} with $r = z'$.  From
 \cref{eq:greedy:window_nqry:domain}, we have $r = z' \le z$, and
 from \cref{eq:greedy:window_nqry:a}, we have $l' \le a(z)$ at this line, and
 again from \cref{eq:greedy:window_nqry:domain}, we have $l' \le
 a'$.  From the similar reasoning, we have $l' \le a'$ in
 \lineref{algo:hot_select:window} of all previous iterations where $r <
 z'$.  Therefore, $l$ has not been updated to a value greater than $a'$
 in \lineref{algo:hot_select:l} up to this point.
 In addition, from \cref{eq:greedy:window_nqry:no_hot}, we have not
 found a hot range in any iteration where $a(z) \le r < z'$, so $l$ has not
 been updated to a value greater than $a'$ in
 \lineref{algo:hot_select:next_scan} either.  So, $l \le a'$ in
 \lineref{algo:hot_select:l} of the iteration with $r = z' \le z$, and
 thus
 \begin{equation}
  \{c_i\}_{i=l}^r \supseteq \{c_i\}_{i=a'}^{z'}.  \label{eq:greedy:window_nqry:supset}
 \end{equation}
 According to \cref{eq:greedy:window_nqry:no_hot}, $\{c_i\}_{i=l}^r$
 should not have been selected as a hot range, so
 $\NQrys\left(\{c_i\}_{i=l}^r\right) < \frac{Q}{P}$ holds in
 \lineref{algo:hot_select:thres}.  Hence,
 \begin{align*}
  \NQrys\left(\{c_i\}_{i=a(z)}^z \setminus h^*\right) &= \NQrys\left(\{c_i\}_{i=a'}^{z'}\right)  \quad (\because \text{\cref{eq:greedy:window_nqry:range}}) \\
  &\le \NQrys\left(\{c_i\}_{i=l}^r\right)  \quad (\because \text{\cref{eq:greedy:window_nqry:supset}}) \\
  &< \frac{Q}{P}.
 \end{align*}
\end{proof}

\begin{theorem}
 \label{thm:greedy:cold_nqry}
 The query load on the ranges not selected as hot ranges in each base
 partition is less than $\alpha \frac{Q}{P}$.
\end{theorem}

\begin{proof}
 Let $b \coloneqq \{c_i\}_{i=s}^e$ be a base partition, $h$ be
 $\FindHot(b)$, $h^*$ be $\bigcup h$.  Then, this theorem says that
 $\NQrys(b \setminus h^*) \le \alpha \frac{Q}{P}$.

 Let $r_1 \coloneqq e$ and
 \[
  l_1 \coloneqq \max\left\{ l \in [s, r_1] \relmiddle| \SizeOf\left( \{c_i\}_{i=l}^{r_1} \right) \ge \frac{1}{\alpha}\frac{D}{P} \right\}
  \]
 be defined.  Note that because
 \[
  \SizeOf\left(\{c_i\}_{i=s}^e\right) = \SizeOf(b) = \frac{D}{P} \ge
 \frac{1}{\alpha}\frac{D}{P}
 \]
 based on \lineref{algo:hot-cold:base} in
 \cref{algo:hot-cold}, such $l_1 \ge s$ certainly exists.  From
 \cref{thm:greedy:window_nqry}, we have
 $\NQrys\left(\{c_i\}_{i=l_1}^{r_1} \setminus h^*\right) < \frac{Q}{P}$.
 Similarly, if $\SizeOf\left(\{c_i\}_{i=s}^{l_1-1}\right) \ge
 \frac{1}{\alpha}\frac{D}{P}$, then $r_2 \coloneqq l_1 - 1$, $l_2
 \coloneqq \max\left\{ l \in [s, r_2] \relmiddle| \SizeOf\left(
 \{c_i\}_{i=l}^{r_2} \right) \ge \frac{1}{\alpha}\frac{D}{P} \right\}$
 can be defined, and $\NQrys\left(\{c_i\}_{i=l_2}^{r_2} \setminus
 h^*\right) < \frac{Q}{P}$ holds.

 Let us suppose that such a truncation from the end of $b$ can be
 performed $n$ times and $\{l_j\}_{j=1}^n, \{r_j\}_{j=1}^n$ are
 obtained.  Then, for $j \in [1, n]$,
 \begin{gather}
  l_j \coloneqq \max\left\{ l \in [s, r_j] \relmiddle| \SizeOf\left( \{c_i\}_{i=l}^{r_j} \right) \ge \frac{1}{\alpha}\frac{D}{P} \right\},  \label{eq:greedy:cold_nqry:cut_size} \\
  \NQrys\left(\{c_i\}_{i=l_j}^{r_j} \setminus h^*\right) < \frac{Q}{P}.  \label{eq:greedy:cold_nqry:cut_nqry}
 \end{gather}
 Since $l_{j+1} \le r_{j+1} = l_j - 1 < l_j$ for $j \in [1, n)$, the
 sequence $\{l_j\}_{j=1}^n$ is strictly monotonically decreasing, and
 thus for sufficiently large $n$,
 $\SizeOf\left(\{c_i\}_{i=s}^{l_n-1}\right) <
 \frac{1}{\alpha}\frac{D}{P}$ holds.

 It can be shown that such $n$ is less than or equal to $\alpha$.
 Suppose that $n$ is greater than $\alpha$.  Then,
 \begin{equation}
  \SizeOf\left(\{c_i\}_{i=s}^{l_\alpha-1}\right) \ge \frac{1}{\alpha}\frac{D}{P}.  \label{eq:greedy:cold_nqry:supposed}
 \end{equation}
 \begin{align*}
  \therefore \SizeOf(b) &= \SizeOf\left(\{c_i\}_{i=s}^{l_\alpha-1}\right) + \SizeOf\left(\{c_i\}_{i=l_\alpha}^e\right) \\
  &\ge \frac{1}{\alpha}\frac{D}{P} + \sum_{j=1}^\alpha \SizeOf\left(\{c_i\}_{i=l_j}^{r_j}\right)  \quad (\because \text{\cref{eq:greedy:cold_nqry:supposed}}) \\
  &\ge \frac{\alpha+1}{\alpha}\frac{D}{P}.  \quad (\because \text{\cref{eq:greedy:cold_nqry:cut_size}})
 \end{align*}
 On the other hand, from \lineref{algo:hot-cold:base} in
 \cref{algo:hot-cold}, we have
 \begin{equation}
  \SizeOf(b) = \frac{D}{P},  \label{eq:greedy:cold_nqry:base_size}
 \end{equation}
 which contradicts the above.

 Therefore, either $n = \alpha$ or $n \le \alpha - 1$ holds, and either
 assumption allows us to show that $\NQrys(b \setminus h^*) \le \alpha
 \frac{Q}{P}$.  From \cref{eq:greedy:cold_nqry:cut_size},
 \begin{equation}
  \SizeOf\left(\{c_i\}_{i=l_n}^e\right) = \sum_{j=1}^n \SizeOf\left(\{c_i\}_{i=l_j}^{r_j}\right) \ge \frac{n}{\alpha} \frac{D}{P}.  \label{eq:greedy:cold_nqry:tail_size}
 \end{equation}
 From \cref{eq:greedy:cold_nqry:cut_nqry},
 \begin{equation}
  \NQrys\left(\{c_i\}_{i=l_n}^e \setminus h^*\right) = \sum_{j=1}^n \NQrys\left(\{c_i\}_{i=l_j}^{r_j} \setminus h^*\right) < n \frac{Q}{P}.  \label{eq:greedy:cold_nqry:tail_nqry}
 \end{equation}
 When $n = \alpha$.  Since
 \begin{align*}
  \SizeOf\left(\{c_i\}_{i=s}^{l_n-1}\right) &= \SizeOf(b) - \SizeOf\left(\{c_i\}_{i=l_n}^e\right) \\
  &\le \frac{D}{P} - \frac{n}{\alpha} \frac{D}{P}  \quad (\because \text{\cref{eq:greedy:cold_nqry:base_size,eq:greedy:cold_nqry:tail_size}}) \\
  &= 0  \quad (\because n = \alpha),
 \end{align*}
 $\{c_i\}_{i=s}^{l_n-1}$ is an empty sequence of data chunks, and so\break
 $\NQrys(\{c_i\}_{i=s}^{l_n-1}) = 0$.  Therefore,
 \begin{align*}
  &\NQrys(b \setminus h^*) \\
  &= \NQrys\left(\{c_i\}_{i=s}^{l_n-1} \setminus h^*\right) + \NQrys\left(\{c_i\}_{i=l_n}^e \setminus h^*\right) \\
  &< 0 + n \frac{Q}{P}  \quad (\because \text{\cref{eq:greedy:cold_nqry:tail_nqry}}) \\
  &= \alpha \frac{Q}{P}.  \quad (\because n = \alpha)
 \end{align*}

 When $n \le \alpha - 1$.  From \cref{thm:greedy:window_nqry}, we have
 \begin{equation}
  \NQrys\left(\{c_i\}_{i=s}^{l_n-1} \setminus h^*\right) < \frac{Q}{P},  \label{eq:greedy:cold_nqry:head_nqry}
 \end{equation}
 and therefore,
 \begin{align*}
  &\NQrys(b \setminus h^*) \\
  &= \NQrys\left(\{c_i\}_{i=s}^{l_n-1} \setminus h^*\right) + \NQrys\left(\{c_i\}_{i=l_n}^e \setminus h^*\right) \\
  &< \frac{Q}{P} + n \frac{Q}{P}  \quad (\because \text{\cref{eq:greedy:cold_nqry:tail_nqry,eq:greedy:cold_nqry:head_nqry}}) \\
  &\le \alpha \frac{Q}{P}.  \quad (\because n \le \alpha - 1)
 \end{align*}
\end{proof}

\begin{theorem}
 The number of hot ranges found in all base partitions is at most $P$
 in total.
\end{theorem}

\begin{proof}
 From \cref{thm:greedy:hot_nqry}, the number of queries targeting each
 hot range is not less than $\frac{Q}{P}$.  Hence, since the total
 number of queries is $Q$, the number of hot ranges is at most
 $\frac{Q}{\frac{Q}{P}} = P$.
\end{proof}

\begin{theorem}
 The data load and query load on each DPU are less than
 $\left(\frac{1}{\alpha} + 1\right)\frac{D}{P} + M_d$ and $(\alpha + 1
 )\frac{Q}{P} + M_q$, respectively.
\end{theorem}

\begin{proof}
 From \cref{thm:greedy:hot_size}, the data load on each hot partition is
 less than $\frac{1}{\alpha}\frac{D}{P} + M_d$.  The data load on the
 remaining cold partitions in each DPU is no more than that of a base
 partition, which is $\frac{D}{P}$ according to
 \lineref{algo:hot-cold:base} of \cref{algo:hot-cold}.  Because
 \cref{algo:hot-cold} assigns at most one hot partition to a single DPU,
 the data load on each DPU is less than
 $\left(\frac{1}{\alpha}\frac{D}{P} + M_d\right) + \frac{D}{P} =
 \left(\frac{1}{\alpha} + 1\right)\frac{D}{P} + M_d$.

 From \cref{thm:greedy:hot_nqry}, the query load on each hot partition
 is less than $\frac{Q}{P} + M_q$.  From \cref{thm:greedy:cold_nqry},
 the query load on the remaining cold partitions in each DPU is less than
 $\alpha\frac{Q}{P}$.  Therefore, the query load on each DPU is less
 than $\left(\frac{Q}{P} + M_q\right) + \alpha\frac{Q}{P} = (\alpha +
 1)\frac{Q}{P} + M_q$.
\end{proof}

\section{Proof of the Properties of the Double-Scan Hot Selection Algorithm}
\label{sec:proof:double}

We prove the properties satisfied by the result of
\cref{algo:double-scan} below.  Unless otherwise specified, line numbers
are used to refer to the processing and states within
\cref{algo:double-scan}.

\begin{theorem}
 \label{thm:double:hot_consecutive}
 Each hot partition consists of a sequence of consecutive data chunks in
 the key space.
\end{theorem}

\begin{proof}
 It is selfevident for hot ranges selected in the first scan in
 \lineref{algo:double-scan:greedy}.  Let $b \coloneqq \{c_i\}_{i=s}^e$ be a base partition, and consider a hot range in $b$ added to $h$
 in \lineref{algo:double-scan:append}.  When we reach
 \lineref{algo:double-scan:append}, $l$ can be expressed using the function $a$
 defined in \cref{eq:greedy:window_nqry:L,eq:greedy:window_nqry:a} as
 \begin{equation}
  l = \max\{l_0, a(r)\}.  \label{eq:double:hot_consecutive:l}
 \end{equation}
 In the proof of \cref{thm:greedy:window_nqry}, it has been shown that,
 for all $z \in [s, e]$,
 $\{c_i\}_{i=a(z)}^z \setminus h^*$ is either $\emptyset$ or a sequence
 of consecutive data chunks.  Therefore, by
 \cref{eq:double:hot_consecutive:l}, $\{c_i\}_{i=l}^r \setminus h^*$
 is also either $\emptyset$ or a sequence of consecutive data chunks.
 If it is $\emptyset$ then it will not be added to $h$ in
 \lineref{eq:double:hot_consecutive:l}, so the set of hot ranges $h$
 consists only of sequences of consecutive data chunks.
\end{proof}

\begin{theorem}
 \label{thm:double:hot_size}
 The data load on each hot partition is less than
 $\frac{1}{\alpha}\frac{D}{P} + M_d$.
\end{theorem}

\begin{proof}
 For the hot ranges selected in the first scan in
 \lineref{algo:double-scan:greedy}, refer to \cref{thm:greedy:hot_size}.
 Below, we consider a hot range selected in the second scan.  By a discussion
 analogous to the proof of \cref{thm:greedy:hot_size}, it can be derived
 that
 \[
  \SizeOf\left( \{c_i\}_{i=l}^r \right) < \frac{1}{\alpha}\frac{D}{P} + M_d
 \]
 holds when we reach \lineref{algo:double-scan:append}.  Therefore,
 \[
  \SizeOf\left( \{c_i\}_{i=l}^r \setminus h^* \right) \le \SizeOf\left( \{c_i\}_{i=l}^r \right) < \frac{1}{\alpha}\frac{D}{P} + M_d.
 \]
 Consequently, the data load of the elements added to $h$ in
 \lineref{algo:double-scan:append} is less than
 $\frac{1}{\alpha}\frac{D}{P} + M_d$.
\end{proof}

\begin{theorem}
 \label{thm:double:hot_nqry}
 The query load on each hot partition is less than
 $\frac{Q}{P} + M_q$.
\end{theorem}

\begin{proof}
 For the hot ranges selected in the first scan in
 \lineref{algo:double-scan:greedy}, refer to \cref{thm:greedy:hot_nqry}.
 For the ranges selected in the second scan, the proposition holds
 because the following inequality holds when we reach \lineref{algo:double-scan:append}.
 \begin{align*}
  \NQrys\left(\{c_i\}_{i=l}^r \setminus h^*\right) &\le \NQrys\left(\{c_i\}_{i=a(r)}^r \setminus h^*\right)  \quad (\because \text{\cref{eq:double:hot_consecutive:l}}) \\
  &< \frac{Q}{P}  \quad (\because \text{\cref{thm:greedy:window_nqry}}) \\
  &< \frac{Q}{P} + M_q.  \quad (\because M_q \ge 0)
 \end{align*}
\end{proof}

\begin{definition}
 \label{thm:double:L}
 Let $r, t_l$ be indices of data chunks such that $t_l \le r$.
 Let $t_s$ be a data size.
 Define the function $L(r, t_l, t_s)$ that returns an index of
 data chunk as follows: The sequence of data chunks
 $\{c_i\}_{i=L(r, t_l, t_s)}^r$ is the shortest consecutive sequence that ends at
 $c_r$, begins with $c_{t_l}$ or later, and contains data of size $t_s$
 or greater.  If no such sequence exists, return $t_l$.  $L$ is
 formalized as
 \[
  L(r, t_l, t_s) \coloneqq \begin{cases}
   t_l & (D(r, t_l, t_s) = \emptyset) \\
   \max D(r, t_l, t_s) & (\text{otherwise})
  \end{cases},
 \]
 where $D(r, t_l, t_s)$ is
 \[
  D(r, t_l, t_s) \coloneqq \left\{ l \in [t_l, s] \relmiddle| \SizeOf\left( \{c_i\}_{i=l}^r \right) \ge t_s \right\}.
 \]
\end{definition}

Using \cref{thm:double:L}, we can express that in \lineref{algo:double-scan:window}, the value of $l$ is updated to
\begin{equation}
 \label{eq:double:l_of_window}
 l = L\left( r, s, \frac{\beta}{\alpha}\frac{D}{P} \right).
\end{equation}
Also, the update of $l$ in \lineref{algo:double-scan:small_window} can
be expressed as
\begin{equation}
 \label{eq:double:l_of_hot}
 l = L\left( r, l_0, \frac{1}{\alpha}\frac{D}{P} \right).
\end{equation}

\begin{theorem}
 \label{thm:double:n_hot}
 The number of hot ranges found in all base partitions is at most $P$ in
 total.
\end{theorem}

\begin{proof}
 Let $b \coloneqq \{c_i\}_{i=s}^e$ be a base partition and consider the
 number of hot ranges found in it.  We define $h_1$ and $h_2$ as
 \begin{align*}
  h_1 &\coloneqq \FindHotGreedy(b), \\
  h_2 &\coloneqq \FindHot(b) \setminus h_1.
 \end{align*}
 Then the number of hot ranges in
 $b$ is $n_h \coloneqq |h_1| + |h_2|$.  From \cref{thm:greedy:hot_nqry},
 \begin{equation}
  \NQrys\left( \bigcup h_1 \right) \ge |h_1| \frac{Q}{P}.  \label{eq:double:n_hot:greedy_nqry}
 \end{equation}

 From here, we show that the following inequality holds:
 \begin{equation}
  \NQrys\left( b \setminus \bigcup h_1 \right) \ge |h_2| \frac{Q}{P}.  \label{eq:double:n_hot:other_nqry}
 \end{equation}
 When we reach \lineref{algo:double-scan:beta}, $\beta$ is set to
 \begin{equation}
  \beta = \left\lfloor \frac{\NQrys\left( b \setminus \bigcup h_1 \right)}{Q/P} \right\rfloor.  \label{eq:double:n_hot:beta}
 \end{equation}
 Therefore, 
 \begin{gather}
  \beta \le \frac{\NQrys\left( b \setminus \bigcup h_1 \right)}{Q/P} \nonumber \\
  \therefore \NQrys\left( b \setminus \bigcup h_1 \right) \ge \beta \frac{Q}{P}.  \label{eq:double:n_hot:other_nqry_beta}
 \end{gather}
 If $\beta = 0$, we will exit \cref{algo:double-scan} in
 \lineref{algo:double-scan:early}, so \cref{eq:double:n_hot:other_nqry}
 holds because $|h_2| = 0$.  For the case where $\beta \ge 1$, we show
 $|h_2| \le \beta$ to produce \cref{eq:double:n_hot:other_nqry} from
 \cref{eq:double:n_hot:other_nqry_beta}.

 Suppose
 \begin{equation}
  |h_2| > \beta.  \label{eq:double:n_hot:too_many}
 \end{equation}
 This means that the addition of a hot range to $h$ was performed $|h_2| (> \beta)$ times in \lineref{algo:double-scan:append}.
 Therefore, if we denote the number of iterations of the loop
 \footnote{$I_{ter}$ can be greater than $|h_2|$. This is because nothing is added to $h$ in the iterations where $\{c_i\}_{i=1}^r \subset h^*$ occurs in \lineref{algo:double-scan:append}.}
 starting at \lineref{algo:double-scan:while} as $I_{ter}$, then
 \begin{equation}
  \label{eq:double:n_hot:iter_gt_beta}
  I_{ter} \ge |h_2| > \beta \ge 1.
 \end{equation}
 Here, we denote the values of $l$ and $r$ when we reach \lineref{algo:double-scan:append} during the $k$-th iteration ($k \in [1, I_{ter}]$) as $l_k$ and $r_k$, respectively.

 From \cref{eq:double:l_of_hot},
 \begin{equation}
  \label{eq:double:n_hot:hot_r_to_l}
  {}^\forall k \in [1, I_{iter}], l_k = L\left( r_k, l_0, \frac{1}{\alpha}\frac{D}{P} \right).
 \end{equation}
 Because $r$ is reinitialized to $r_0$ immediately before the loop begins (\lineref{algo:double-scan:init_r}),
 \begin{equation}
  \label{eq:double:n_hot:r_1}
  r_1 = r_0.
 \end{equation}
 Based on the update of $r$ in \lineref{algo:double-scan:next_r},
 \begin{equation}
  \label{eq:double:n_hot:hot_l_to_r}
  {}^\forall k \in [2, I_{ter}], r_k = l_{k-1} - 1.
 \end{equation}
 From \cref{thm:double:L} and \cref{eq:double:n_hot:hot_r_to_l}, we have
 \begin{gather}
  {}^\forall k \in [1, I_{ter}], l_k \le r_k, \text{and} \label{eq:double:n_hot:l_le_r} \\
  {}^\forall k \in [1, I_{ter}], l_k \ge l_0. \label{eq:double:lbound_lk}
 \end{gather}
 From \cref{eq:double:lbound_lk}, in \cref{algo:double-scan:next_r}, $r$ is updated to a value of $l_0 - 1$ or greater.
 Consequently, if the condition in \cref{algo:double-scan:while} immediately afterward becomes false, $r$ can only be $l_0 - 1$.
 Thus, the value of $l$ in the final iteration is $l_0$, meaning
 \begin{equation}
  \label{eq:double:n_hot:l_Iter}
  l_{I_{ter}} = l_0.
 \end{equation}
 Combining \cref{eq:double:n_hot:r_1,eq:double:n_hot:hot_l_to_r,eq:double:n_hot:l_le_r,eq:double:n_hot:l_Iter}, we have
 \begin{equation}
  \label{eq:double:n_hot:order_of_l_r}
  l_0 = l_{I_{ter}} \le r_{I_{ter}} < l_{I_{ter}-1} \le r_{I_{ter}-1} < \dots < l_1 \le r_1 = r_0
 \end{equation}

 Consider the data size of $\{c_i\}_{i=l_k}^{r_0}$ for $k \in [1, I_{ter} - 1]$.
 For all $k$ within this range, \cref{eq:double:n_hot:order_of_l_r} implies that $l_k > l_0$.
 So, from \cref{eq:double:n_hot:hot_r_to_l},
 \begin{equation}
  \label{eq:double:n_hot:size_per_iter}
  {}^\forall k \in [1, I_{ter} - 1], \SizeOf\left( \{c_i\}_{i=l_k}^{r_k} \right) \ge \frac{1}{\alpha}\frac{D}{P}
 \end{equation}
 holds. Then, for all $k \in [1, I_{ter} - 1]$, we have
 \begin{align*}
  \SizeOf\left( \{c_i\}_{i=l_k}^{r_0} \right)
  &= \sum_{j=1}^k \SizeOf\left( \{c_i\}_{i=l_j}^{r_j} \right) \quad (\because \text{\cref{eq:double:n_hot:order_of_l_r}}) \\
  &\ge \sum_{j=1}^k \frac{1}{\alpha}\frac{D}{P} \quad (\because \text{\cref{eq:double:n_hot:size_per_iter}})\\
  &= \frac{k}{\alpha}\frac{D}{P}.
 \end{align*}
 Therefore, from \cref{thm:double:L},
 \[
  {}^\forall k \in [1, I_{ter} - 1], l_k \le L\left( r_0, s, \frac{k}{\alpha}\frac{D}{P} \right).
 \]
 Because we have $\beta \in [1, I_{ter} - 1]$ from \cref{eq:double:n_hot:iter_gt_beta}, this implies
 \begin{align}
  l_\beta &\le L\left( r_0, s, \frac{\beta}{\alpha}\frac{D}{P} \right) \nonumber \\
  \therefore l_0 = l_{I_{iter}} < l_\beta &\le L\left( r_0, s, \frac{\beta}{\alpha}\frac{D}{P} \right). \quad (\because \text{\cref{eq:double:n_hot:iter_gt_beta,eq:double:n_hot:order_of_l_r}})
  \label{eq:double:n_hot:l0_ubound}
 \end{align}
 However, because \cref{eq:double:l_of_window} holds whenever we reach \lineref{algo:double-scan:update_max}, we have
 \[
  l_0 = L\left( r_0, s, \frac{\beta}{\alpha}\frac{D}{P} \right),
 \]
 which contradicts \cref{eq:double:n_hot:l0_ubound}.
 Therefore,
 \begin{equation}
  \label{eq:double:n_hot:n_warm}
  |h_2| \le \beta.
 \end{equation}

 From \cref{eq:double:n_hot:n_warm,eq:double:n_hot:other_nqry_beta},
 \cref{eq:double:n_hot:other_nqry} holds.  Then, from
 \cref{eq:double:n_hot:greedy_nqry,eq:double:n_hot:other_nqry},
 \begin{align*}
  \NQrys(b) &= \NQrys\left( \bigcup h_1 \right) + \NQrys\left( b \setminus \bigcup h_1 \right) \\
  &\ge |h_1| \frac{Q}{P} + |h_2| \frac{Q}{P} \\
  &= n_h \frac{Q}{P}.
 \end{align*}
 Because this holds for all base partitions, if we denote the total
 number of hot ranges as $n$, then
 \[
  Q \ge n \frac{Q}{P},
 \]
 proving that $n \le P$.
\end{proof}

\begin{theorem}
 \label{thm:double:cold_nqry}
 The query load on the ranges not selected as hot ranges in each base
 partition is at most $\frac{Q}{P} \max\left\{\frac{\alpha + 1}{3}, 1\right\}$.
\end{theorem}

\begin{proof}
Let us define a base partition $b$, the set of hot ranges selected in the first and second scans $h_1$ and $h_2$ in the same manner as in the proof of \cref{thm:double:n_hot}.
\begin{align*}
    b &\coloneqq \{c_i\}_{i=s}^e\ ,\\
    h_1 &\coloneqq \FindHotGreedy(b), \text{and}\\
    h_2 &\coloneqq \FindHot(b) \setminus h_1.\\
\end{align*}
Also, let us define
\begin{align*}
  q &\coloneqq \NQrys\left( b \setminus \bigcup h_1 \right), \text{and}\\
  \Delta q &\coloneqq \NQrys\left( \bigcup h_2 \right).
\end{align*}
This theorem says that $q - \Delta q \le \frac{Q}{P} \max\left\{\frac{\alpha + 1}{3}, 1\right\}$.

By the same discussion as in the proof of \cref{thm:double:n_hot}, we have
\begin{align}
  \beta & = \left\lfloor \frac{\NQrys\left( b \setminus \bigcup h_1 \right)}{Q/P} \right\rfloor \nonumber\\ 
        & = \left\lfloor \frac{q}{Q/P}\right\rfloor.  \label{eq:double:cold_nqry:beta}
\end{align}
 From \cref{thm:greedy:cold_nqry}, we have $q < \alpha \frac{Q}{P}$,
 then $0 \le \beta < \alpha$.
 If $\beta = 0$, we have
 \begin{align*}
  q - \Delta q &= q \\
  &< \frac{Q}{P}  \quad (\because \text{\cref{eq:double:cold_nqry:beta}}) \\
  &\le \frac{Q}{P} \max\left\{\frac{\alpha + 1}{3}, 1\right\}.
 \end{align*}
 Below, we prove $q - \Delta q \le \frac{\alpha + 1}{3} \frac{Q}{P}$ for
 $1 \le \beta < \alpha$. Note that $\frac{\alpha + 1}{3}$ is not less
 than $1$ since we have $\alpha \ge \beta + 1 \ge 2$ in this case.

 Consider the case where $\frac{2\alpha - 1}{3} \le \beta < \alpha$. Let
 $l_e$ be $L\left( e, s, \frac{\beta}{\alpha} \frac{D}{P} \right)$. From
 \cref{eq:double:l_of_window}, \cref{algo:double-scan} for the base
 partition $b$ examines $\{c_i\}_{i=l_e}^e$ as a candidate of
 $\{c_i\}_{i=l_0}^{r_0}$ in \lineref{algo:double-scan:cmp_nqry}.
 Because the algorithm finally selects $\{c_i\}_{i=l_0}^{r_0}$ for which
 the remaining query count $\NQrys\left( \{c_i\}_{i=l_0}^{r_0} \right)$
 reaches its maximum value $\Delta q$,
 \[
  \NQrys\left( \{c_i\}_{i=l_e}^e \setminus \bigcup h_1 \right) \le \Delta q.
 \]
 If $l_e = s$, since this left-hand side is the same as $q$, we have
 $q \le \Delta q$, then
 $q - \Delta q \le 0 \le \frac{\alpha + 1}{3} \frac{Q}{P}$. Otherwise,
 \begin{align}
  q - \Delta q &\le \NQrys\left( b \setminus \bigcup h_1 \right) - \NQrys\left( \{c_i\}_{i=l_e}^e \setminus \bigcup h_1 \right) \nonumber \\
  &= \NQrys\left( \{c_i\}_{i=s}^{l_e - 1} \setminus \bigcup h_1 \right).  \label{eq:double:cold_nqry:big_beta:step1}
 \end{align}
 Let us define
 \[
  n \coloneqq \left\lceil \frac{\SizeOf\left( \{c_i\}_{i=s}^{l_e - 1} \right)}{\frac{1}{\alpha} \frac{D}{P}} \right\rceil.
 \]
 Then, we can prove that $\{c_i\}_{i=s}^{l_e - 1}$ can be covered by $n'$
 disjoint ranges $\{c_i\}_{i=l_j}^{r_j}$ ($j \in [1, n']$) such that
 \begin{align*}
  n' &\le n \\
  r_1 &= l_e - 1, \\
  r_j &= l_{j-1} - 1 \text{ for } j \in [2, n'], \\
  l_j &= L\left( r_j, s, \frac{1}{\alpha} \frac{D}{P} \right) \text{ for } j \in [1, n'], \text{and} \\
  l_1 &= s
 \end{align*}
 by a similar discussion as in the proof of \cref{thm:greedy:cold_nqry}.
 From \cref{thm:greedy:window_nqry}, it follows for each $j \in [1, n']$
 that
 \[
  \NQrys\left( \{c_i\}_{i=l_j}^{r_j} \setminus \bigcup h_1 \right) < \frac{Q}{P}.
 \]
 Therefore, from \cref{eq:double:cold_nqry:big_beta:step1},
 \begin{align}
  q - \Delta q &\le \NQrys\left( \{c_i\}_{i=s}^{l_e - 1} \setminus \bigcup h_1 \right) \nonumber \\
  &= \sum_{j=1}^{n'} \NQrys\left( \{c_i\}_{i=l_j}^{r_j} \setminus \bigcup h_1 \right) \nonumber \\
  &< n' \frac{Q}{P} \nonumber \\
  &\le n \frac{Q}{P} \nonumber \\
  &= \left\lceil \frac{\SizeOf\left( \{c_i\}_{i=s}^{l_e - 1} \right)}{\frac{1}{\alpha} \frac{D}{P}} \right\rceil \frac{Q}{P}.  \label{eq:double:cold_nqry:big_beta:step2}
 \end{align}
 From the definitions of $l_e$ and $L$,
 \begin{align*}
  \SizeOf\left( \{c_i\}_{i=s}^{l_e - 1} \right) &= \SizeOf(b) - \SizeOf\left( \{c_i\}_{i=l_e}^r \right) \\
  &\le \frac{D}{P} - \frac{\beta}{\alpha} \frac{D}{P} \\
  &= (\alpha - \beta) \frac{1}{\alpha} \frac{D}{P}.
 \end{align*}
 Then, from \cref{eq:double:cold_nqry:big_beta:step2},
 \begin{align*}
  q - \Delta q &\le \lceil \alpha - \beta \rceil \frac{Q}{P} \\
  &= (\alpha - \beta) \frac{Q}{P} \quad (\because \alpha, \beta \in \mathbb{N}) \\
  &\le \left( \alpha - \frac{2\alpha - 1}{3} \right) \frac{Q}{P} \quad \left(\because \beta \ge \frac{2\alpha - 1}{3} \right) \\
  &= \frac{\alpha + 1}{3} \frac{Q}{P}.
 \end{align*}

 We have proved the theorem for $\beta = 0$ and
 $\frac{2\alpha - 1}{3} \le \beta < \alpha$. From here, we forcus on the
 case where $1 \le \beta < \frac{2\alpha - 1}{3}$. Let
 \begin{equation}
  m \coloneqq \left\lceil \frac{\SizeOf(b)}{\frac{\beta}{\alpha}\frac{D}{P}} \right\rceil.  \label{eq:double:cold_nqry:n}
 \end{equation}
 Then, we can prove that the base partition $b$ can be covered by $m'$ disjoint ranges $\{c_i\}_{i=l_j'}^{r_j'}$ ($j \in [1, m']$) such that
\begin{align*}
  m' &\le m, \\
  r_1' &= e, \\
  r_j' &= l_{j-1}' - 1 \text{ for } j \in [2, m'], \\
  l_j' &= L\left( r_j', s, \frac{\beta}{\alpha} \frac{D}{P} \right) \text{ for } j \in [1, m'], \text{and} \\
  l_1' &= s
\end{align*}
  by a similar discussion as in the proof of \cref{thm:greedy:cold_nqry}.
The for-loop on \lineref{algo:double-scan:for} in \cref{algo:double-scan} examines ranges $\{c_i\}_{i=l}^{r}$ including all such ranges and selects the one with the maximum number of remaining queries
\[
\NQrys\left(\{c_i\}_{i=l}^{r} \setminus \bigcup h_1\right).
\]
Because hot ranges in $h_2$ are created by dividing the selected range in this for-loop, $\Delta q = \NQrys\left( \bigcup h_2 \right)$ is the number of queries for the selected range.
Thus,
 \begin{equation}
  {}^\forall j \in [1, m'], \NQrys\left(\{c_i\}_{i=l_j'}^{r_j'} \setminus \bigcup h_1\right) \le \Delta q
 \end{equation}
holds.

 Therefore,
 \begin{align*}
  q &= \NQrys\left(b \setminus \bigcup h_1\right) \\
  &= \sum_{j=1}^{m'} \NQrys\left(\{c_i\}_{i=l_j'}^{r_j'} \setminus \bigcup h_1\right) \\
  &\le m' \Delta q \\
  &\le m \Delta q.
 \end{align*}
 Then, 
 \begin{align}
  q - \Delta q &\le q - \frac{q}{m} \nonumber \\
  &= q\left(1 - \frac{1}{m}\right) \nonumber \\
  &= q\left(1 - \frac{1}{\left\lceil \frac{\SizeOf(b)}{\frac{\beta}{\alpha}\frac{D}{P}} \right\rceil}\right) \quad (\because \text{\cref{eq:double:cold_nqry:n}}) \nonumber \\
  &= q\left(1 - \frac{1}{\left\lceil \frac{\alpha}{\beta} \right\rceil}\right)  \label{eq:double:cold_nqry:mid_beta:step1}
 \end{align}
 Here, let $r \coloneqq \frac{q}{Q/P} - \beta$. We have
 $r \in [0, 1)$ by \cref{eq:double:cold_nqry:beta}.  Using this
 to evaluate the right-hand side of \cref{eq:double:cold_nqry:mid_beta:step1},
 \[
  q - \Delta q \le (\beta + r) \frac{Q}{P} \left(1 - \frac{1}{\left\lceil \frac{\alpha}{\beta} \right\rceil}\right).
 \]
 Because this right-hand side is monotonically increasing with respect
 to $r$, from $r < 1$, we have
 \begin{equation}
  q - \Delta q < (\beta + 1) \frac{Q}{P} \left(1 - \frac{1}{\left\lceil \frac{\alpha}{\beta} \right\rceil}\right).  \label{eq:double:cold_nqry:mid_beta:step2}
 \end{equation}

 Below, we will consider two different cases:
 $1 \le \beta < \frac{\alpha}{2}$ and
 $\frac{\alpha}{2} \le \beta < \frac{2\alpha - 1}{3}$. When
 $1 \le \beta < \frac{\alpha}{2}$, we have
 $\left\lceil \frac{\alpha}{\beta} \right\rceil \ge 3$. Therefore,
 \begin{align*}
  q - \Delta q &< (\beta + 1) \frac{Q}{P} \left(1 - \frac{1}{\left\lceil \frac{\alpha}{\beta} \right\rceil}\right) \quad (\because \text{\cref{eq:double:cold_nqry:mid_beta:step2}}) \\
  &\le (\beta + 1) \frac{Q}{P} \left(1 - \frac{1}{3}\right).
 \end{align*}
 Because $\beta < \frac{\alpha}{2}$ implies
 $\beta \le \frac{\alpha - 1}{2}$ when $\alpha$ and $\beta$ are
 integers,
 \begin{align*}
  \therefore q - \Delta q &\le \left( \frac{\alpha - 1}{2} + 1 \right) \frac{Q}{P} \left(1 - \frac{1}{3}\right) \\
  &= \frac{\alpha + 1}{3} \frac{Q}{P}.
 \end{align*}

 When $\frac{\alpha}{2} \le \beta < \frac{2\alpha - 1}{3}$, we have
 $\left\lceil \frac{\alpha}{\beta} \right\rceil = 2$. Therefore,
 \begin{align*}
  q - \Delta q &< (\beta + 1) \frac{Q}{P} \left(1 - \frac{1}{\left\lceil \frac{\alpha}{\beta} \right\rceil}\right) \quad (\because \text{\cref{eq:double:cold_nqry:mid_beta:step2}}) \\
  &= (\beta + 1) \frac{Q}{P} \left(1 - \frac{1}{2}\right) \quad \left(\because \left\lceil \frac{\alpha}{\beta} \right\rceil = 2 \right) \\
  &< \left( \frac{2\alpha - 1}{3} + 1 \right) \frac{Q}{P} \left(1 - \frac{1}{2}\right) \quad \left(\because \beta < \frac{2\alpha - 1}{3} \right) \\
  &= \frac{\alpha + 1}{3} \frac{Q}{P}.
 \end{align*}
\end{proof}
\end{Proof-in-Appendix}

\end{document}